\pdfoutput=1
\documentclass[conference]{IEEEtran}
\IEEEoverridecommandlockouts
\usepackage{hyperref}
\usepackage{booktabs} 
\usepackage{graphicx}
\usepackage{amsfonts}
\usepackage{graphics, tikz}
\usepackage{subfig,stfloats}
\definecolor{storeClusterComponent}{HTML}{FDAE61}
\definecolor{dbscan}{HTML}{ABDDA4}
\usepackage{pdfpages}
\usepackage{enumitem}
\usepackage{csquotes}
\usepackage{amsmath}
\usepackage{multirow}
\usepackage{bbm}
\usepackage{balance}
\usepackage{accents}
\newcommand{\nop}[1]{}
\usepackage[font=bf,labelfont=bf]{caption}
\usepackage{url}
\usepackage{tikz,pgfplots,pgfplotstable}
\usetikzlibrary{pgfplots.groupplots}
\usetikzlibrary{positioning, arrows, shapes}
\usepackage{algorithm}
\usepackage[noend]{algpseudocode}

\newtheorem{thm}{Theorem}[section]

\newtheorem{defn}{Definition}
\newtheorem{lem}{Lemma}
\newtheorem{pro}{Proposition}
\newtheorem{cor}{Corollary}

\newtheorem{examp}{Example}

\usepackage{wasysym}

\usepackage{amssymb}
\usepackage{pifont}
\usepackage{multirow}

\pgfplotsset{compat=1.11,
    /pgfplots/ybar legend/.style={
    /pgfplots/legend image code/.code={%
       \draw[##1,/tikz/.cd,yshift=-0.25em]
        (0cm,0cm) rectangle (0.8em,4pt);},
   },
}

\usepackage{cite}

\usepackage{textcomp}
\usepackage{xcolor}
\def\BibTeX{{\rm B\kern-.05em{\sc i\kern-.025em b}\kern-.08em
    T\kern-.1667em\lower.7ex\hbox{E}\kern-.125emX}}

\def\l{\boldsymbol\ell}
\def\g{\mbox{\bf g}}
\def\w{\mbox{\bf w}}
\def\-{\mbox{-}}

\definecolor{lightblue}{rgb}{0.68, 0.85, 0.9}
	\definecolor{lavenderrose}{rgb}{0.98, 0.63, 0.89}
	\definecolor{lavender(web)}{rgb}{0.9, 0.9, 0.98}
		\definecolor{junglegreen}{rgb}{0.16, 0.67, 0.53}
			\definecolor{grannysmithapple}{rgb}{0.66, 0.89, 0.63}
			\definecolor{lightcoral}{rgb}{0.94, 0.5, 0.5}
\definecolor{applegreen}{rgb}{0.55, 0.71, 0.0}

\def\G{\mathcal G}
\def\H{\mathcal H}
\def\X{\mathcal X}
\def\C{\mathcal C}
\def\S{\mathcal S}

\def\t{\bigtriangleup}
\def\Pr{\text{Pr}}

\begin{document}
\title{Nucleus Decomposition in Probabilistic Graphs: \\ Hardness and Algorithms}

\author{Fatemeh Esfahani, Venkatesh Srinivasan, Alex Thomo, and Kui Wu
\IEEEcompsocitemizethanks{\IEEEcompsocthanksitem F. Esfahani,V. Srinivasan, A. Thomo and K. Wu are with the Department
of Computer Science, University of Victoria, Victoria,
B.C.\protect\\
E-mail: {esfahani,srinivas,thomo,wkui}@uvic.ca. 
}}


\maketitle

\begin{abstract}
\looseness=-1
Finding dense components in graphs is of great importance in analysing the structure of networks. 
Popular and computationally feasible frameworks for discovering dense subgraphs are core and truss decompositions. 
Recently, Sarıyüce~et~al. introduced nucleus decomposition, which uses $r$-cliques contained in $s$-cliques, where $s>r$, as the basis for defining dense subgraphs. Nucleus decomposition can reveal interesting subgraphs that can be missed by core and truss decompositions.

In this paper, we present {\em nucleus decomposition in probabilistic graphs}.
The major questions we address are: How to define meaningfully nucleus decomposition in probabilistic graphs? How hard is computing nucleus decomposition in probabilistic graphs? Can we devise efficient algorithms for exact or approximate nucleus decomposition in large graphs?

We present three natural definitions of nucleus decomposition in probabilistic graphs: {\em local}, {\em global}, and {\em weakly-global}. 
\looseness=-1
We show that the local version is in PTIME, whereas global and weakly-global are \#P-hard and NP-hard, respectively. 
We present an efficient and exact dynamic programming approach for the local case. Further, we present statistical approximations that can scale to bigger datasets without much loss of accuracy.   
For global and weakly-global decompositions we complement our intractability results by proposing efficient algorithms that give approximate solutions based on search space pruning and Monte-Carlo sampling. 
Extensive experiments show the scalability and efficiency of our
algorithms. Compared to probabilistic core and truss decompositions, nucleus decomposition significantly outperforms in terms of density and clustering metrics.
\end{abstract}


\section{Introduction}
Probabilistic graphs are graphs where each edge has a probability of existence (cf.~\cite{bonchi2014core,mukherjee2015mining,jin2011discovering,kempe2003maximizing,budak2011limiting,Tang2014,jin2011distance,zou2010discovering}). Many real-world graphs, such as social, trust, and biological networks are associated with intrinsic uncertainty. 
For instance, 
in social and trust networks, an edge can be weighted by the probability of influence or trust between two users that the edge connects~\cite{goyal2010learning,korovaiko2013trust,kuter2010using}.
In biological networks of protein-protein interactions (cf.~\cite{Genome}) an edge can be assigned a probability value representing the strength of prediction that a pair of proteins will interact in a living organism~\cite{dittrich2008identifying, dong2007understanding, sharan2007network}. 

Mining dense subgraphs and discovering hierarchical relations among them is a fundamental problem in graph analysis tasks. 
For instance, it can be used in 
visualizing complex networks \cite{zhao2012large}, 
finding correlated genes and motifs in biological networks \cite{zhang2005general,fratkin2006motifcut},
detecting communities in social and web graphs~\cite{fang2020survey,li2020finding}, summarizing text~\cite{antiqueira2009complex}, and 
revealing new research subjects in citation networks~\cite{sariyuce2018local}.
Core and truss decompositions are popular tools for finding dense subgraphs. 
A $k$-core is a maximal subgraph in which each vertex has at least $k$ neighbors, and a $k$-truss is a maximal subgraph whose edges are contained in at least $k$ triangles.
Core and truss decompositions have been extensively studied for deterministic as well as probabilistic graphs (cf.~\cite{bonchi2014core,esfahani2019efficient,huang2016truss,khaouid2015k,peng2018efficient,wang2012truss}). 

A recent notion of dense subgraphs is {\em nucleus} introduced by Sarıyüce~et~al.~\cite{sariyuce2015finding,sariyuce2017nucleus}. 
Nucleus decomposition is a generalization of core and truss decompositions that uses higher-order structures to detect dense regions. 
It can reveal interesting subgraphs that can be missed by core and truss decompositions. 
In a nutshell, a $k$-$(r,s)$-nucleus is a maximal subgraph whose $r$-cliques are contained in at least $k$ of $s$-cliques, where $s>r$. 
For $r=1,s=2$ and $r=2,s=3$ we obtain the notions of $k$-core and $k$-truss, respectively.    
For $r=3,s=4$, 
$r$-cliques are {\em triangles},
$s$-cliques are 4-{\em cliques}, and  
$k$-$(3,4)$-nucleus is strictly stronger than $k$-truss and $k$-core. 
Sarıyüce~et~al. in ~\cite{sariyuce2015finding,sariyuce2017nucleus} observed that, in practice, $k$-$(3,4)$-nucleus is the most interesting in terms of the quality of subgraphs produced for a large variety of graphs. As such, in this paper we also focus on this decomposition. 
%
%
%
To the best of our knowledge, \textit{nucleus decomposition over probabilistic graphs} has not been studied yet.

As pointed out by \cite{sariyuce2015finding,sariyuce2017nucleus}, nucleus decomposition can uncover a finer grained structure of dense groups not possible using other dense subgraph mining methods; as such, nucleus decomposition can be beneficial for a large variety of applications, e.g. 
community structure discovery~\cite{saxena2018social}, 
mining dense regions in internet of things~\cite{zhao2019effective}, 
financial fraud detection~\cite{zhang2017hidden},
extracting brain connectome subgraph hierarchy \cite{wu2020extracting}, 
detection of complexes in biological networks~\cite{ma2017detection}, etc.
%
All these applications of nucleus decomposition 
extend naturally to the probabilistic networks. 
Ignoring probabilities and using deterministic methods amounts to setting all probabilities to 1, which not only misses salient information, but could prove detrimental in applications such as 
finding cohesive subnetworks of proteins from probabilistic PPI networks which has valuable implications to disease diagnosis~\cite{dittrich2008identifying}. 
Last but not the least, computing probabilistic nucleus is highly beneficial for task driven team formation in probabilistic social networks, demonstrated later in our case study using a DBLP network.


\subsection{Contributions}

We are the first to study nucleus decomposition in probabilistic graphs.
The major questions we address are: How to define meaningfully nucleus decomposition in probabilistic graphs? How hard is computing nucleus decomposition in probabilistic graphs? Can we devise efficient algorithms for exact or approximate nucleus decomposition in large graphs? 

\smallskip
\noindent
\textbf{Definitions.} 
We start by introducing three natural notions of probabilistic nucleus decomposition (Section~\ref{secpnuclei}). 
They are based on the concept of {\em possible worlds} (PW's), which are instantiations of a probabilistic graph obtained by flipping a biased coin for each edge independently, according to its probability.
We define {\em local}, {\em global}, and {\em weakly-global} notions of nucleus as a maximal probabilistic subgraph $\H$ satisfying different structural conditions for each case. 

In the local case, we require a good number of PW's of $\H$ to satisfy a high level of density around each triangle (in terms of 4-cliques containing it) in $\H$. This is local in nature because the triangles are considered independently of each other. 
To contrast this, we introduce the global notion, where we request the PW's themselves be deterministic nuclei. This way, not only do we achieve density around each triangle but also ensure the same is achieved for all the triangles of $\H$ simultaneously. 
Finally, we relax this strict requirement for the weakly-global case by requiring that PW's only contain a deterministic nucleus that includes the triangles of $\H$.


\smallskip
\noindent
\textbf{Global and Weakly-Global Cases.}
We show that computing global and weakly-global decompositions are intractable, namely \#P-hard and NP-hard, resp. (Section~\ref{hardnessSection}). 
We complement these results with efficient algorithms for these two cases that give approximate solutions 
based on search space pruning combined with Monte-Carlo sampling (Section~\ref{globalnucleus}). 


\smallskip
\noindent
\textbf{Local Case.}
We show that local nucleus decomposition is in PTIME (Section~\ref{localnucleus}).
The main challenge is to compute the probability of each triangle to be contained in $k$ 4-cliques. We present a dynamic programming (DP) solution for this task, which combined with a triangle peeling approach, solves the problem of local nucleus decomposition efficiently.
While this is welcome result, we further propose statistical methods to speed-up the computation. Namely, we provide a framework where well-known distributions, such as Poisson, Normal, and Binomial, can be employed to approximate the DP results. 
We provide detailed conditions under which the approximations can be used reliably, otherwise DP is used as fallback. 
This hybrid approach speeds-up the computation significantly and is able to handle datasets, which DP alone cannot.

\smallskip
\noindent
\textbf{Experiments.}
We present extensive experiments which show that 
our DP method for local nucleus decomposition is efficient and can handle large datasets; 
when combined with our statistical approximations, the process is significantly sped-up and can handle much larger datasets. 
%
We demonstrate the importance of nucleus decomposition by comparing it to probabilistic core and truss decomposition using density and clustering metrics. 
The results show that nucleus decomposition significantly outperforms core and truss decompositions in terms of these metrics.


\section{Deterministic Nuclei}\label{background}


Let $G=(V,E)$ be an undirected graph, where $V$ is a set of vertices, and $E$ is a set of edges. 
For a vertex $v \in V$, let $N(v)$ be the set of $v$'s neighbors: $N(v) = \left \{ u : (u,v) \in E \right \}$. 
The (deterministic) degree of $v$ in $G$, is equal to $\left| N(v) \right|$.

\smallskip
\noindent
\textbf{Nucleus decomposition in deterministic graphs.} Nucleus decomposition is a generalization of core and truss decompositions~\cite{sariyuce2015finding,sariyuce2017nucleus}. 
Each nucleus is a subgraph which contains a dense cluster of cliques. The formal definitions are as follows.

Let $r,s$ with $r < s$ be positive integers.
We call cliques of size $r$, {\em $r$-cliques}, and denote them by $R, R'$, etc.
Analogously, we call cliques of size $s$, {\em $s$-cliques}, and denote them by $S, S'$, etc.

\begin{defn}
The {\em $s$-support} of an $r$-clique $R$ in $G$, denoted by $s$-$supp_G(R)$, is the number of $s$-cliques in $G$ that contain $R$.
\end{defn}

\begin{defn}
Two $r$-cliques $R$ and $R'$ in $G$, are {\em $s$-connected}, if there exists a sequence $R = R_1,R_2,\cdots,R_k=R'$ of $r$-cliques in $G$ such that for each $i$, there exists some $s$-clique in $G$ that contains $R_i \cup R_{i+1}$.
\end{defn}


Now nucleus decomposition is as follows. 

\begin{defn} 
Let $k$ be a positive integer. 
A \pmb{$k$-$(r,s)$}\textbf{-nucleus} is a maximal subgraph $H$ of $G$ with the following properties.
\begin{enumerate}
    \item $H$ is a union of $s$-cliques: every edge in $H$ is part of an $s$-clique in $H$.
    \item $s$-$supp_H(R)\geq k$ for each $r$-clique $R$ in $H$.
    \item Each pair $R,R'$ of $r$-cliques in $H$ is $s$-connected in $H$.
\end{enumerate}
\end{defn}


For simplicity, whenever clear from the context, we will drop the use of prefix $s$ from the definition of support and connectedness. 

\looseness=-1
When $r=1,s=2$, 
$r$-cliques are nodes, $s$-cliques are edges, and $k$-$(1,2)$-nucleus is the well-known notion of $k$-core. 
When $r=2,s=3$,  
$r$-cliques are edges,
$s$-cliques are triangles, and  $k$-$(2,3)$-nucleus is the well-known notion of $k$-truss.
\cite{sariyuce2015finding} shows that $k$-$(3,4)$-nucleus, where  we consider triangles contained in 4-cliques, provides much more interesting insights compared to $k$-core and $k$-truss in terms of density and hierarchical structure. As such, in this paper, we also focus on the $r=3,s=4$ case. 
For simplicity, we will drop using $r$ and $s$ and assume them to be 3 and 4, respectively. In particular, we will refer to  
$k$-$(3,4)$-nucleus as simply $k$-nucleus.

\section{Probabilistic Nuclei}\label{secpnuclei}

\noindent
\looseness=-1
\textbf{Probabilistic Graphs.}
A probabilistic graph is a triple $\G=(V,E,p)$, where $V$ and $E$ are as before and $p : E \rightarrow (0,1]$ is a function that maps each edge $e \in E$ to its existence probability $p_e$.
In the most common probabilistic model (cf. \cite{bonchi2014core,jin2011discovering,kempe2003maximizing}), the existence probability of each edge is assumed to be independent of other edges. 

In order to analyze probabilistic graphs, we use the concept of \textit{possible worlds} that are deterministic graph instances of $\G$ in which only a subset of edges appears. 
Conceptually, the possible worlds are obtained by flipping a biased coin for each edge independently, according to its probability.
We write $G \sqsubseteq  \G$ to say that $G$ is possible world for $\G$.
The probability of a possible world $G=(V,E_G) \sqsubseteq  \G$ is as follows:
$\Pr[G\mid\G] = \prod_{e \in E_G} p_e \prod_{e \in E\setminus E_G}(1-p_e).$

We will use $\G$, $\G'$, $\H$, $\H'$ to denote probabilistic graphs.

\smallskip
\noindent
\textbf{Nucleus decomposition in probabilistic graphs.}
We now define three variants of nucleus decomposition in probabilistic graphs which are based on Definitions \ref{lwg} and \ref{lwgnucleus} we give below. 
These variants relate to the nature of nucleus and we refer to them as {\bf local~($\boldsymbol\ell$)}, {\bf global~(g)}, and {\bf weakly-global~(w)}.     

\begin{defn}\label{lwg}
Let $\mathcal H$ be a probabilistic graph, 
$\t$ a triangle, and 
$\mu$ a mode in set $\{\l, \g, \w\}$. 
Then, $X_{\H,\t,\mu}$ is a random variable that takes integer values $k$ with tail probability
\begin{equation}\label{eq1}
     \Pr(X_{\H,\t,\mu} \geq k) = \sum_{H \sqsubseteq  \H} \Pr[H \mid \H] \cdot \mathbbm{1}_{\mu}(H,\t,k),
\end{equation}
where indicator variable $\mathbbm{1}_{\mu}(H,\t,k)$
is defined depending on mode $\mu$ as follows. 

\begin{description}
    \item [$\mathbbm{1}_{\mbox{$\boldsymbol\ell$}}(H,\t,k)=1$] if $\t$ is in $H$, and the support of $\t$ in $H$ is at least~$k$. 

    \item [$\mathbbm{1}_{\mbox{\bf g}}(H,\t,k)=1$] if $\t$ is in $H$, and $H$ is a deterministic $k$-nucleus.     
    
    \item [$\mathbbm{1}_{\mbox{\bf w}}(H,\t,k)=1$] if $\t$ is in $H$, and there is a subgraph $H'$ of $H$ that contains $\t$ and is a deterministic $k$-nucleus.

\end{description}
\end{defn}

It is clear that
$(\mathbbm{1}_{\mbox{\bf g}}(H,\t,k)=1)$ 
$\implies$
$(\mathbbm{1}_{\mbox{\bf w}}(H,\t,k)=1)$
$\implies$
$(\mathbbm{1}_{\mbox{$\boldsymbol\ell$}}(H,\t,k)=1)$.

\looseness=-1
In the above definition, $\mathbbm{1}_{\l}(H,\t,k)$ has a local quality because a possible world $G$ satisfies its condition if it provides sufficient support to triangle $\t$ without considering other triangles in $H$.
On the other hand,  $\mathbbm{1}_{\g}(H,\t,k)$ and $\mathbbm{1}_{\w}(H,\t,k)$ have a global quality because 
a possible world $H$ satisfies their conditions only when other triangles in $H$ are considered as well (creating a nucleus together).



In the following, as {\em preconditions} for cohesiveness, we will assume {\em cliqueness} and {\em connectedness} for the nuclei subgraphs we define. 
Specifically, we will only consider subgraphs $\H$, which, ignoring edge probabilities, are unions of 4-cliques, and where each pair of triangles in $\H$ is connected in $\H$. 



\begin{defn}\label{lwgnucleus}
Let $\G=(V,E,p)$ be a probabilistic graph. Given threshold $\theta \in [0,1]$, integer $k\geq 0$, and $\mu \in \{\l, \g, \w\}$, 
a $\mu\mbox{-}(k,\theta)$-nucleus $\H$ is a maximal subgraph of $\G$, such that  
$\Pr(X_{\H,\t,\mu} \geq k) \geq \theta$
for each triangle $\t$ in $\H$.

Moreover, the {\em $\mu\mbox{-}(k,\theta)$- nucleusness} (or simply nucleusness when $\mu$, $k$, and $\theta$ are clear from context) of a triangle $\t$ is the largest value of $k$ such that $\t$ is contained in a $\mu\mbox{-}(k,\theta)$-nucleus.


\end{defn}


Intuitively for $\mu=\l$, from a probabilistic perspective, a subgraph $\mathcal H$ of $\mathcal G$ can be regarded as a cohesive
subgraph of $\mathcal G$ if the support of every triangle in $\H$ is no less than $k$ with high probability (no less than a threshold $\theta$).
We call this version local nucleus. 

Local nucleus is a nice concept for probabilistic subgraph cohesiveness, however, it has the following shortcoming.  
While it ensures that 
every triangle $\t$ in $\mathcal H$ has support at least $k$ in a good number of instantiations of $\mathcal H$, it does not ensure 
those instantiations are deterministic nuclei themselves or
they contain some nucleus which in turn contains $\t$. Obviously, nucleusness is a desirable property to ask for in order to achieve a higher degree of cohesiveness and this leads to the other two versions of probabilistic nucleus of a global nature, which we call global and weakly-global (obtained for $\mu=\g$ and $\mu=\w$).

In general, $\g$-$(k,\theta)$-nuclei are smaller and more cohesive than $\w$-$(k,\theta)$-nuclei. 
We remark that, every $\g$-$(k,\theta)$-nucleus is contained in a $\w$-$(k,\theta)$-nucleus which in turn is contained in an  $\l$-$(k,\theta)$-nucleus.

\begin{figure}
    \centering
   \includegraphics[scale=0.265]{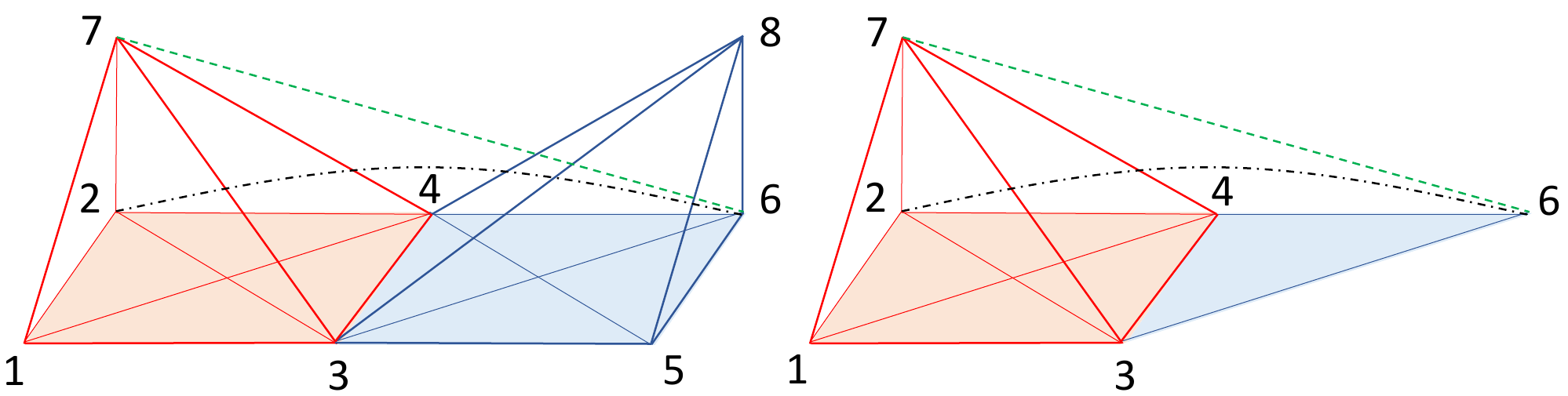}
    \caption{[Left] Probabilistic graph $\G$. Red edges have probability $P=0.9$, blue edges have probability $P'=0.8$, the green dashed edge has probability $1$, and 
the black dot-dashed edge has probability $P''=0.5$. [Right] Subgraph $\H$ which is $\w$-$(2,0.13)$ nucleus. $\mathcal{H}_1$ and $\mathcal{H}_2$ induced by $\{ 1,2,3,4,7\}$ and $\{ 2,3,4,6,7\}$ are $\g$-$(2,0.13)$ nuclei.}
    \label{fig:my_label}
    \vspace{-0.7cm}
\end{figure}

\begin{examp}
\looseness=-1
Consider graph $\G$ shown in Figure~\ref{fig:my_label} [Left].
Let us assume that 
the red edges have probability $P=0.9$, the blue edges have probability $P'=0.8$, the green dashed edge has probability $1$, and 
the black dot-dashed edge has probability $P''=0.5$. 
Let $\theta=0.13$.
It can be verified that each triangle $\bigtriangleup$ in $\mathcal{G}$ is contained in at least $2$ $4$-cliques with probability at least $0.134$, i.e. $\Pr(X_{\G,\t,\l} \geq 2) \geq 0.134 \geq \theta$. Thus, $\G$ is a $\l$-$(2,0.13)$ nucleus. 

\looseness=-1
However, $\G$ cannot be a $\w$-$(2,0.13)$ or $\g$-$(2,0.13)$ nucleus. For instance, consider triangle $\bigtriangleup=(3,5,6)$. In all the possible worlds of $\G$, the clique on vertices $\{ 3,4,5,6,8\}$ should exist since this is the only deterministic $2$-nucleus which contains $\bigtriangleup$. Thus, the edges of this clique (9 blue and 1 red) should exist and the other edges in $\G$ can either exist or not in the possible worlds of $\G$. As a result, we get $\Pr(X_{\G,\t,\w} \geq 2) = 0.8^9\cdot 0.9 = 0.120 < \theta$.


\smallskip
Now, consider subgraph $\mathcal{H}$ induced by vertices $\{1,2,3,4,6,7 \}$, Figure~\ref{fig:my_label} [Right]. 
We show that $\H$ is $\w$-$(2,0.13)$-nucleus. Our reasoning is as follows. 
Ignoring probabilities, $\H$ consists of two deterministic $2$-nuclei, one induced by
$\{1,2,3,4,7\}$ (call it $cl_1$) and the other induced by $\{2,3,4,6,7\} $ (call it $cl_2$).
Triangles in $\H$ can belong to either $cl_1$ or $cl_2$. 
Consider an arbitrary triangle $\bigtriangleup$ in $cl_1$. To compute $\Pr(X_{\H,\t,\w} \geq 2)$, all the possible worlds of $\H$ which contain $cl_1$ as a deterministic $2$-nucleus are valid. 
As a result, all the edges in $cl_1$ should exist, and the edges $(2,6)$, $(3,6)$ and $(4,6)$ can either exist or not exists in the valid possible worlds (edge $(7,6)$ has probability 1). As such, we have $2^3=8$ valid possible worlds. Summing over the existence probability of each possible world, we get $\Pr(X_{\H,\t,\w} \geq 2) = 0.9^{10} = 0.348 > \theta$. A similar reasoning can be applied for an arbitrary triangle $\bigtriangleup'$ in $cl_2$ which gives $\Pr(X_{\H,\t',\w} \geq 2) = 1\cdot 0.5 \cdot 0.8^2 \cdot 0.9^6 = 0.170 > \theta$. Thus, we can say that $\H$ is a $\w$-$(2,0.13)$-nucleus.

Let us consider $\H$ in more detail. This subgraph cannot be a  $\g$-$(2,0.13)$-nucleus. For instance, consider triangle 
$\bigtriangleup = (1,2,3)$. 
For this triangle, there are only two valid possible worlds which are deterministic $2$-nucleus: $(1)$ the one in which all the edges exist $(H_1)$, 
$(2)$ the one in which none of edges $(2,6)$, $(3,6)$ and $(4,6)$ exist $(H_2)$. Adding one of the edges $(2,6)$, $(3,6)$ and $(4,6)$ creates one extra triangle which will not be part of two cliques. This results in the possible world not being a deterministic $2$-nucleus. So, summing over these two possible worlds we get: 
\begin{align*}
    & \Pr(X_{\H,\t,\g} \geq 2)  = 
    0.9^{10}\cdot 1 \cdot 0.5 \cdot 0.8^2   \\
    & + 
    0.9^{10}\cdot 1 \cdot (1-0.5) \cdot (1-0.8)^2 
    = 0.118 < \theta. 
\end{align*}

However consider subgraphs $\H_1$ and $\H_2$ induced by $\{1,2,3,4,7\}$ and $\{2,3,4,6,7\}$, respectively. 
The only possible worlds of $\H_1$ and $\H_2$ that are deterministic $2$-nuclei are the ones in which all their edges exist. 
So for each triangle $\bigtriangleup$ and $\bigtriangleup'$ in $\H_1$ and $\H_2$, we have  $\Pr(X_{\H_1,\t,\g} \geq 2) = 0.9^{10} = 0.348 > \theta$ and $\Pr(X_{\H_2,\t',\g} \geq 2) = 1\cdot 0.5 \cdot 0.8^2 \cdot 0.9^6 = 0.170 > \theta$. Thus, $\mathcal{H}_1$ and $\mathcal{H}_2$ are global $\g$-$(2,0.13)$-nuclei subgraphs.
\end{examp}

\noindent
{\bf Nucleus Decomposition.}
The {\em nucleus decomposition} finds the set of all the $\mu\-(k,\theta)$-nuclei for different values of $k$.
We will study the problem in the three different modes we consider. Specifically, we call nucleus-decomposition problems for the different modes 
{$\l$}-{\sl NuDecomp}, 
{$\g$}-{\sl NuDecomp}, and
{$\w$}-{\sl NuDecomp}, respectively.

{\color{black}
In the following, we prove uniqueness and hierarchical-containment properties of probabilistic nucleus decomposition.
\begin{pro}
The local, weakly-global, and global nucleus decompositions are unique.
\begin{IEEEproof}
The uniqueness is based on the definitions of local, weakly-global, and global nucleus decomposition. Specifically, the uniqueness follows from the property that each nucleus (local, weakly-global, or global) is a {\em maximal} subgraph satisfying the required property. As such, the set of maximal nuclei is unique.
\end{IEEEproof}
\end{pro}

\begin{pro}
There exists a hierarchical-containment property for local, weakly-global, and global decomposition.
\begin{IEEEproof}
Let $\theta$ be an arbitrary and fixed user-defined threshold. To prove the hierarchical-containment property for local nucleus decomposition, let $\mathcal{F}$ be a local $\l$-$(k+1,\theta)$-nucleus. 
By the definition of local nucleus, each triangle in $\mathcal{F}$ has support at least $k+1$ in $\mathcal{F}$, with probability no less than $\theta$. Since $k+1 > k$, each triangle in $\mathcal{F}$ has also support at least $k$ in $\mathcal{F}$. 
Thus, $\mathcal{F}$ is contained in a $\l$-$(k,\theta)$-nucleus. 
This proves the property for local nucleus decomposition.

For weakly-global decomposition, let $\mathcal{H}$ be a weakly-global $\w$-$(k+1,\theta)$-nucleus. Referring to Definition~\ref{lwg} for each triangle $\bigtriangleup \in \mathcal{H}$ we have

\begin{equation}\label{eq1revision}
     \Pr(X_{\H,\t,\w} \geq k+1) = \sum_{H \sqsubseteq  \H} \Pr[H \mid \H] \cdot \mathbbm{1}_{\w}(H,\t,k+1) \geq \theta ,
\end{equation}
where $\mathbbm{1}_{\mbox{\bf w}}(H,\t,k+1)=1$ means that $\t$ is in $H$, and there is a subgraph $H'$ of $H$ that contains $\t$ and is a {\em deterministic} $(k+1)$-nucleus in $H$. 
Since every deterministic $(k+1)$-nucleus is contained in a deterministic $k$-nucleus \cite{sariyuce2015finding}, 
we have that there exist a deterministic $k$-nucleus $H''$ in $H$ that contains $H'$. Clearly, $H''$ contains $\t$, and we have that
$\mathbbm{1}_{\mbox{\bf w}}(H,\t,k+1)=1$ implies 
$\mathbbm{1}_{\mbox{\bf w}}(H,\t,k)=1$, thus
$\Pr(X_{\H,\t,\w} \geq k) \geq \theta$, i.e. 
$\mathcal{H}$ is contained in a $\w$-$(k,\theta)$-nucleus. 

A similar reasoning holds for the global case as its definition is based on possible worlds which are deterministic nuclei.
\end{IEEEproof}
\end{pro}
}
 
We show that {$\l$}-{\sl NuDecomp} can be computed in polynomial time and furthermore we give several algorithms to achieve efficiency for large graphs. 
Before this, we start by showing that 
{$\g$}-{\sl NuDecomp} and {$\w$}-{\sl NuDecomp} are $\#P$-hard and NP-hard, respectively.
Nevertheless, as we show later in the paper, once we obtain the {$\l$}-{\sl NuDecomp}, we can use it as basis, combined with sampling techniques, to effectively approximate {$\g$}-{\sl NuDecomp} and {$\w$}-{\sl NuDecomp}.


\section{Hardness Results}\label{hardnessSection}

In this section, we show that  {$\g$}-{\sl NuDecomp}  and {$\w$}-{\sl NuDecomp} are NP-hard. For this we use a reduction from the $k$-clique problem. 
Furthermore, we can show that {$\g$}-{\sl NuDecomp} is even harder, namely $\#$P-hard, using a reduction from the network reliability problem.  

\begin{defn}
    \textbf{The $k$-clique Problem \cite{jin2011discovering}.} Given a graph $G$, and input parameter $k$, the $k$-clique problem is to check whether there is a clique of size $k$ in the graph. The $k$-clique problem is NP-complete.
    \end{defn}


We note the following interesting property about $k$-nucleus.

\begin{lem}\label{lemm2}
\normalfont For any $k$, the only graph on $(k+3)$ vertices which is a deterministic $k$-nucleus is a $(k+3)$-clique.
\end{lem}
\begin{IEEEproof}
Recall that based on the definition of the $k$-nucleus, each triangle is contained in at least $k$ $4$-cliques.  Given the vertices $\left \{ v_1,v_2, \cdots , v_{k+3} \right \}$, without loss of generality, let $\t_{123} = (v_1,v_2,v_3)$ be a triangle with vertices $v_1$,$v_2$, and $v_3$. The triangle $\t_{123}$ must be part of $k$ 4-cliques; therefore, there must be an edge between each of the remaining $k$ vertices and all the vertices of the triangle $\t_{123}$. Now, new triangles are created, containing vertices $\left \{ v_4, \cdots, v_{k+3} \right \}$. Let $\t_{ijt}$ be one of them, where $i,j = 1,2$, $i \neq j$, and $t=4,\cdots,k+3$. This triangle must be contained in $k$ $4$-cliques as well. Thus, there should be edges between each vertex in the triangle $\t_{ijt}$ to the other $k$ vertices. Thus, each vertex $v_t$ becomes connected to all the other vertices creating a clique on $k+3$ vertices. 
 \end{IEEEproof}

\begin{thm} \label{hardness2}
    {\bf w}-{\sl NuDecomp} and {\bf g}-{\sl NuDecomp} are NP-hard.
\end{thm}
    \begin{IEEEproof}
    Given a graph $G=(V,E)$, we define a probabilistic graph $\G=(V,E,p)$ as follows: For each edge $e$ in $\G$, $p(e) = \frac{1}{2^{2m}+1}$, where $m$ is the number of the edges in $\G$. 
    Let $\theta = \left(\frac{1}{2^{2m}+1}\right)^{(k+3)\cdot(k+2)/2}$. 
    
    We prove that $\w\mbox{-}(k,\theta)$-nucleus \big($\g\mbox{-}(k,\theta)$-nucleus \big) 
    of $\G$ exists if and only if a $(k+3)$-clique exists in $G$.  Let \textit{C} be a $(k+3)$-clique in $G$. Since \textit{C} has $(k+3)\cdot(k+2)/2$ edges, its existence probability is $\left(\frac{1}{2^{2m}+1}\right)^{(k+3)\cdot(k+2)/2} = \theta$. In addition, in a $(k+3)$-clique, each triangle is contained in exactly $k$ $4$-cliques. Thus, as a subgraph, \textit{C} is a both $\w\mbox{-}(k,\theta)$-nucleus and $\g\mbox{-}(k,\theta)$-nucleus
    of $\G$. 

    In the following we show that if $G$ does not contain a $(k+3)$-clique, the $\w\mbox{-}(k,\theta)$-nucleus and $\g\mbox{-}(k,\theta)$-nucleus are empty. We prove the case for weakly-global, and the same reasoning can be applied for the global case as well.
    
    Suppose that $G$ does not contain a $(k+3)$-clique. For a contradiction, let us assume that a $\w\mbox{-}(k,\theta)$-nucleus of 
    $\G$ exists and denote it by $\H$.
    Based on Lemma \ref{lemm2}, a $(k+3)$-clique is the only graph which has $k+3$ vertices and is a deterministic $k$-nucleus. Since, ignoring edge probabilities, $\H$ cannot be a $(k+3)$-clique, it must have at least $(k+4)$ vertices. Furthermore, since it contains a $k$-nucleus for each triangle in it, the degree of each vertex 
    is at least $(k+2)$.
    
    Let $\t \in \H$ and let $ \left \{  {H}_1,{H}_2,\cdots,H_l \right \}$ be a set of all the valid possible worlds of $\H$, i.e. $\mathbbm{1}_{\mbox{\bf w}}(H_i,\t,k)=1$ ($\mathbbm{1}_{\mbox{\bf g}}(H_i,\t,k)=1$ for the proof of the global case), for all $H_i$ (refer to Definition~\ref{lwg}).  
     The maximum value for $l$ can be $2^m$. For each ${H}_i$, $\Pr(H_i) \leq p(e) ^{(k+4)\cdot(k+2)/2} $. Therefore, for triangle $\t \in \H$, 
    $\Pr(X_{\H,\t,\w} \geq k)$
     is at most:
    $\beta = l \cdot  \left(\frac{1}{2^{2m}+1}\right) ^{(k+4)\cdot(k+2)/2} $. Thus, we have \\
    $\beta \leq 2^m \cdot \left( \frac{1}{2^{2m}+1} \right) ^{ ((k+4)\cdot(k+2))/2} 
         = 2^m \cdot \theta \cdot \left( \frac{1}{2^{2m}+1} \right) ^{(k+2)/2} \\
         < 2^m \cdot \theta \cdot \left( \frac{1}{2^{2m}+1} \right)  < \theta. 
    $
    Thus, $\w\mbox{-}(k,\theta)$-nucleus is empty.
    \end{IEEEproof}
  
  \nop{  
    In what follows, we show that $(k,\theta)$-$(r,r+1)$-g$_2$-nucleus 
    is NP-hard for any integer values of $r$.
    \begin{cor}
    $(k,\theta)$-$(r,r+1)$-g$_2$-nucleus 
    of $\G$ is not empty if and only if $(k+r)$-clique exists in $G$. 
    \begin{IEEEproof}
    \normalfont The proof is similar to the Theorem \ref{hardness2}. The only observations required for the proof are: \textbf{(1)} The degree of each vertex in a $k$-$(r,r+1)$-nucleus, is at least $k+(r-1)$, \textbf{(2)} the only subgraph on $(k+r)$ which is $k$-$(r,r+1)$-nucleus and has $k+r$ vertices is $k+r$-clique. 
    \end{IEEEproof}
    \end{cor}
}

In the extended version of this paper \cite{esfahani2020nucleus}, we show that that {$\g$}-{\sl NuDecomp} is even harder, namely $\#$P-hard.

\section{Local Nucleus decomposition}
\label{localnucleus}

Here we propose efficient algorithms for solving {$\l$}-{\sl NuDecomp}.
Peeling is a general strategy that has been used broadly in core and truss decompositions as well as in deterministic nucleus decomposition~\cite{sariyuce2015finding}.
However, generalizing peeling to compute 
{$\l$}-{\sl NuDecomp} creates significant computational challenges.
For example, 
a challenge is finding the support score for each triangle. 
This is because of the combinatorial nature of finding the maximum value of $k$ such that $\Pr(X_{{\G},\t,\l} \geq k) \geq \theta$ for a triangle $\t$. 
In particular, triangle $\t$ in a probabilistic graph can be part of different numbers of $4$ cliques with different probabilities. As a result, considering all the subsets of $4$-cliques which contain $\t$ results in exponential time complexity. 
In our algorithm, we
identify two challenging tasks, namely computing and updating
nucleus scores.

\subsection{Computing initial nucleus scores}

Our process starts by computing a nucleus score $\kappa_{\t}$ for each triangle $\t$, which initially is the maximum $k$ for which $\Pr(X_{{\G},\t,\l} \geq k) \geq \theta$.

Given a probabilistic graph $\G=(V,E,p)$, let $\t = (u,v,w)$ be a triangle in $\G$. 
For $i=1,\ldots,c_{\t}$, where $c_{\t}= | N(u) \cap N(v) \cap N(w) |$,
let 
$z_i \in N(u) \cap N(v) \cap N(w)$ and 
$S_i = \{ u,v,w, z_i\}$. 
In other words, for each $i$, $S_i$ is the set of vertices of a $4$-clique that contains $\t$. 
For notational simplicity, we will also denote by $S_i$ the 4-clique on $\{u,v,w,z_i\}$.

\looseness=-1
Similarly, for each $i$,
let $\mathcal{E}_i = \left \{(u,z_i),(v,z_i),(w,z_i)  \right \}$ be the set of edges which connect vertex $z_i$ to vertices of $\t$. Let $\Pr(\mathcal{E}_i) = p(u,z_i) \cdot p(v,z_i) \cdot p(w,z_i)$ be the existence probability of $\mathcal{E}_i$.
We have:
\begin{equation}\label{eq2}
    \Pr(X_{\G,\t,\l} \geq k) = \Pr(X_{\G,\t,\l} \geq k-1)- \Pr(X_{\G,\t,\l} = k-1)
\end{equation}

Thus,
we need to compute $\Pr(X_{\G,\t,\l} = k)$ for any $k$, and find the maximum value of $k$ for which the probability on the left-hand side of Equation \ref{eq2} is greater than or equal to $\theta$. { \color{black} In fact, $\Pr(X_{\G,\t,\l} = k)$ gives the probability that $\bigtriangleup$ is contained in $k$ number of $4$-cliques in $\mathcal{G}$. Under the condition that $\bigtriangleup$ exists, we denote $\mathcal{X}(S_\bigtriangleup, k,j)$ to be the probability that $\t$ is contained in $k$ of $4$-cliques from $\{ S_1,\cdots, S_j \} \subseteq {\mathcal S}_{\t}$, where ${\mathcal S}_{\t}$ the set of 4-cliques containing $\t$ in $\mathcal{G}$. In other words, $\mathcal{X}(S_\bigtriangleup, k,j)$ is conditional probability (conditioning on the existence of $\bigtriangleup$).
}

We fix an arbitrary order on ${\mathcal S}_{\t}$.  
The event that 
$\t$ is contained in $k$ of $4$-cliques from $\{ S_1,\cdots, S_j \}$, 
can be expressed as the union of the following two sub-events: 
\textbf{(1)}~the event that the $4$-clique $S_j$ exists and $\t$ is contained in  
$(k-1)$ of $4$-cliques from $\{ S_1,\cdots, S_{j-1} \}$, and 
\textbf{(2)}~the event that the $S_j$ does not exist and $\t$ is part of $k$ of $4$-cliques from $\{ S_1,\cdots, S_{j-1}\}$.
Thus, we have the following recursive formula: 
\begin{align}\label{recursive}
   \X({\mathcal S}_{\t}, k,j) & = \Pr(\mathcal{E}_j) \cdot \X({\mathcal S}_{\t}, k-1,j-1) \\ \nonumber
   & + (1-\Pr(\mathcal{E}_j)) \cdot \mathcal{X}({\mathcal S}_{\t},k,j-1), 
\end{align}
where $k \in [0,c_{\t}]$, and $j \in [0,c_{\t}]$.
Initially, we set 
$\mathcal{X}({\mathcal S}_{\t},0,0) = 1$,
$\mathcal{X}({\mathcal S}_{\t},-1,j) = 0$ for any $j$,
and $\mathcal{X}({\mathcal S}_{\t},k,j) = 0$, if $k > j$. Setting $j = c_{\t}$ in Equation~\ref{recursive}, and multiplying $\mathcal{X}({\mathcal S}_{\t},k,j)$ by $\Pr(\t)$ (existence probability of $\t$), gives the desired probability  $\Pr(X_{\G,\t,\l} = k)$. 
{\color{black}
Thus, we have:
\begin{equation}
    \Pr(X_{\G,\t,\l} = k) = \Pr(\t) \cdot \mathcal{X}(S_\bigtriangleup, k,c_{\t}),
\end{equation}
}

Given a triangle $\t$, let the {\em neighbor triangles} of $\t$ be those triangles which form a $4$-clique with $\t$. In the following we show how we can update $\Pr(X_{\G,\t,\l} \geq k)$ when a neighbor triangle is processed in the decomposition.



\subsection{Updating nucleus scores}

Once the $\kappa$ scores have been initialized as described above, a process of peeling ``removes'' the triangle $\t^*$ of the lowest $\kappa$-score, specifically marks it as {\em processed}, and updates the neighboring triangles $\t$  (those contained in the same 4-cliques as the removed triangle) in terms of $\Pr(X_{\G,\t,\l} \geq k)$.   
Because of the removal of $\t^*$ the cliques containing it cease to exist, thus $\Pr(X_{\G,\t,\l} \geq k)$ of the neighbors $\t$ will change.  
We recompute this probability using the formula in Equation~\ref{recursive}, where the sets of cliques ${\mathcal S}_{\t}$ are updated to remove the cliques containing $\t^*$.

\begin{algorithm}
\caption{{$\l$}-{\sl NuDecomp}}
\label{peelingnuclei}
\begin{algorithmic}[1] 
\Function{$\l$-Nucleusness}{$\G$, $\theta$}
\ForAll{triangles $ \t \in \G$} \label{newfortriangles}
\State{$\kappa(\t) \gets \arg\max_k \{ \X(\S_{\t}, k, c_{\t}) \geq \theta\}$} \label{kappainitial}
\State{\textit{processed}[$\t$]$\gets$ \textbf{false}} \label{processedornot}
\EndFor
\ForAll{unprocessed $ \t \in \G $ with minimum $\kappa(\t)$} \label{startiter}
\State{$\nu(\t) \gets \kappa(\t)$} \label{setkappa}
\State{Find set $\S_{\t}$ of $4$-cliques containing $ \t$} \label{cliquefinding}
\ForAll{$S \in \S_{\t}$ with non-processed triangles}
\ForAll{${\t}' \subset S$, ${\t}' \neq \t$, $\kappa(\t') > \kappa(\t)$ } \label{timecomplex1}
\State{$\!\!\!\!\kappa({\t}') \gets$ $ \arg\max_k \{ \X(\S_{\t'}\setminus S, k, c_{\t'}-1) \geq \theta\}$} \label{updateparttime}
\EndFor
\EndFor
\State{\textit{processed}[$\t$]$\gets$ \textbf{true}} \label{enditer}
\EndFor
\State \textbf{return} \textit{array} $\nu(\cdot)$ \label{nucleusscore}
\EndFunction
\end{algorithmic}
\end{algorithm}

\smallskip
\noindent
Algorithm~\ref{peelingnuclei} computes the nucleusness of each triangle in $\G$.
In line~\ref{kappainitial}, for each triangle $\t$, $\kappa({\t})$ is initialized using Equation~\ref{recursive}.
Array \textit{processed} records whether a triangle has been processed or not in the algorithm (line~\ref{processedornot}). 
At each iteration (line~\ref{startiter}-\ref{enditer}), an unprocessed triangle $\t$ with minimum $\kappa({\t})$ is considered, and its nucleus score is set and stored in array $\nu$ (line~\ref{setkappa}). 
Then, the $\kappa(\t')$ values of all the neighboring triangles $\t'$ are updated using Equation~\ref{recursive}. The affected triangles are those unprocessed triangles which are part of the same $4$-clique with triangle $\t$.
The algorithm continues until all the triangles are processed. At the end, each triangle obtains its nucleus score and array $\nu$ with these scores is returned (line~\ref{nucleusscore}). Once all the nucleus scores are obtained, we build $\l$-$(k,\theta)$-nuclei for each value of $k$.



\looseness=-1
Observe that the $\kappa$ values for each triangle at each iteration decrease or stay the same. 
This implies that $\kappa$ for each triangle $\t$ is a monotonic property function similar to properties described in \cite{batagelj2011fast} for vertices. Now, we can use a reasoning similar to the one in \cite{batagelj2011fast} to show that our algorithm, which repeatedly removes a triangle with the smallest $\kappa$ value, gives the correct nucleusness for each triangle.

\medskip
\noindent
\textbf{Time complexity:} 
Using dynamic programming, Lines~\ref{newfortriangles}-\ref{kappainitial} take $O\left(\sum_{\t \in \G} \kappa_{\t} \cdot c_{\t}\right)$, where $\kappa_{\t}$ is the nucleusness obtained for each triangle $\bigtriangleup$ in line~\ref{kappainitial}.
Let $\kappa_{\max}$ be the maximum $\kappa_{\t}$ over all the triangles in $\G$. 
Since $c_{\t} \in 
O(d(u)+d(v)+d(w)) \subseteq O(d_{\max})$, running time of line~\ref{kappainitial} is
$
  O\left(\sum_{\t \in \G} \kappa_{\t} \cdot c_{\t}\right) =  O\left(\sum_{\t \in \G} \kappa_{\max} \cdot d_{\max}\right)=  O\left(\kappa_{\max}  d_{\max}  T_{\G}\right),
$
%
%
where $T_{\G}$ is the total number of triangles in the graph, and $d_{\max}$ is the maximum degree in $\G$. 
For each triangle $\t$, finding all $\S_{\t}$'s in line~\ref{cliquefinding}, takes $O\left(d(u)+d(v)+d(w)\right) = O(d_{\max})$. In addition, lines ~\ref{timecomplex1}-\ref{updateparttime} take
    $O\left( \sum_{\t' \in N(\t)} (\kappa_{\t'} \cdot c_{\t'})\right)$ time,
where $N(\t)$ is the triangles which form a $4$-clique with $\t$. Note that $N(\t) = O(c_{\t})$. Therefore, the running time for processing all the triangles is
$
    O\left(\sum_{\t \in \G} \Big( d(u)+d(v)+d(w) + \sum_{\t' \in N(\t)}  \kappa_{\t'} \cdot c_{\t'} \Big) \right)$
   $  =O\left(\kappa_{\max} d_{\max}^2 T_{\G}\right).
$


Thus, the total running time of Algorithm~\ref{peelingnuclei} is bounded by $O\left(\kappa_{\max} d_{\max}^2 T_{\G}\right)$, and we can state the following. 
    \begin{thm}
   {$\l$}-{\sl NuDecomp} can be computed in polynomial time.
    \end{thm}
\nop{
\begin{IEEEproof}
\textbf{Initialization.} Given a triangle $\t$, the time complexity of finding maximum $k$ such that $ \Pr(X_{\G,\t,\l} \geq k) \geq \theta$ using dynamic programming approach is $O(c_{\t}^{2})$.
Also, for each triangle we need to find the cliques containing it.
However, this is $\sum_{v \in \t}d(v)$ and thus absorbed by the time needed for dynamic programming.

The probability values can be updated in $O(c_{\t}^2)$.
This time can be brought down to linear in $c_{\t}$ by storing the intermediate probability values computed during initialization, however, we do not do that because of excessive memory requirement of $O(T\cdot c_{\t})$, where $T$ is the number of triangles in the graph. 
\end{IEEEproof}
}

The space complexity is $O(T_{\G})$. This space is needed to store triangles (not $4$-cliques) and their $\kappa$ values. 
This is the same as the space complexity of deterministic nucleus decomposition. 

%
While being able to compute {$\l$}-{\sl NuDecomp} in polynomial time is good news, finding the maximum $k$ such that $\Pr(X_{\G,\t,l} \geq k) \geq \theta$ is quadratic in $c_{\t}$ which is not efficient for large probabilistic graphs.
As an alternative approach, 
we will now propose efficient methods to approximate $\Pr(X_{\G,\t,l} \geq k)$ 
in $O(c_{\t})$ time such that the results are practically distinguishable from the exact values. The approximation is based on limit theorems, such as Le Cam's Poisson Limit Theorem \cite{le1960approximation} and Lyapunov's Central Limit Theorem \cite{Lyapunov-Nouvelle}.

\subsection{Approximating $\kappa$ scores}
%

\medskip
\noindent
{\bf Framework.}
Given a triangle $\t = (u,v,w)$, let $S_{i} = \left\{ u,v,w, z_i \right\}$ for $i=1,\ldots,c_{\t}$, as before. 
Also, let $\mathcal{E}_i = \left \{(u,z_i),(v,z_i),(w,z_i)  \right \}$ be the edges that connect $z_i$ to the vertices of $\t$. 

With slight abuse of notation, we also define each $\mathcal{E}_i$ as an indicator random variable which takes on $1$, if all the edges in $\mathcal{E}_i$ exist, and takes on $0$, if at least one of the edges in the set does not exist.
We observe that 
the variables $\mathcal{E}_i$ are mutually independent since 
the sets $\mathcal{E}_i$ do not share any edge.  
 Also, each Bernoulli variable $\mathcal{E}_i$ takes value $1$  with probability  $p(u,z_i) \cdot p(v,z_i) \cdot p(w,z_i)$ and $0$ with $1-(p(u,z_i) \cdot p(v,z_i)) \cdot p(w,z_i))$. 

Let $ \zeta = \sum_{i=1}^{c_\t}\mathcal{E}_i$. 
We can verify the following proposition. 
\begin{pro}
$\Pr(X_{\G,\t,\l} \geq k) = \Pr(\t) \cdot \Pr[\zeta \geq k]$.
\end{pro}

The expectation and variance of $\zeta$ are $ \mu = \sum_{i=1}^{c_\t}  \Pr(\mathcal{E}_i)$ and $ \sigma ^2 = \sum_{i=1}^{c_\t} \big( \Pr(\mathcal{E}_i)  \cdot (1-\Pr(\mathcal{E}_i) \big)$, respectively. 
Now we show that we can approximate the distribution of $\zeta$ using Le Cam's Theorem which makes use of Poisson Distribution \cite{le1960approximation}.


\noindent
\textbf{Poisson Distribution \cite{haight1967handbook}:} A discrete random variable $X$ is said to have Poisson distribution with positive parameter $\lambda$, if the probability mass function of $X$ is given by: 
\begin{equation}
   \Pr[X=k] = \frac{\lambda^ke^{-k}}{k!}, \quad k=0,1,\cdots,
\end{equation}

The expected value of a Poisson random variable is $\lambda$. Setting $\lambda$ to $\mu$, we can approximate the distribution of $\zeta$ by the Poisson distribution. 
Using Le Cam's Theorem \cite{le1960approximation}, the error bound on the approximation is as follows:
\begin{equation}\label{errpoisson}
    \sum_{k=0}^{c_{\t}} \left | \Pr(\zeta = k) - \frac{\lambda ^k e^{-\lambda }}{k!} \right | < 2 \sum_{i=1}^{c_{\t}} \big(\Pr(\mathcal{E}_i) \big)^2 = 2(\mu-\sigma^2).
\end{equation}

Equation~\ref{errpoisson} shows that the Poisson distribution is reliable if $\Pr(\mathcal{E}_i)$ and $c_{\t}$ are small.

We observe that computing tail probabilities for the Poisson distribution is easy in practice as these probabilities satisfy a simple recursive relation. 
\begin{align}\label{poissform}
 & \Pr[\zeta < k]  \approx \sum_{j<k}\frac{e^{-\lambda } \lambda ^j}{j!} =  \sum_{j<k-1}\frac{e^{-\lambda } \lambda ^j}{j!} + \frac{e^{-\lambda } \lambda ^{k-1}}{(k-1)!} \nonumber \\
  &= \Pr[\zeta < k-1] + \frac{\lambda}{k-1}\Pr[\zeta = k-2]
\end{align}
%
%
with base case $\Pr[\zeta < 1] = \Pr[\zeta = 0] = e^{-\lambda}$.
Using Equation \ref{poissform}, and iterating over all values of $k$ from $1$ to $c_\t$, we can evaluate each term $\Pr[\zeta \geq k] = 1-\Pr[\zeta < k]$ in constant time, and find the maximum $k$ such that $\Pr(\t) \cdot \Pr[\zeta \geq k] \geq \theta$. Thus, the time complexity of obtaining $\Pr(X_{\G,\t,\l} \geq k)$ is $O(c_\t)$.


In some applications, $\sum_{i=1}^{c_{\t}} \big(\Pr(\mathcal{E}_i) \big)^2$ in Equation \ref{errpoisson} can be large, even if each $\Pr({\mathcal E}_i)$ is small. As a result, the difference between the variance 
$\sigma^2 = \sum_{i=1}^{c_{\t}} \Pr(\mathcal{E}_i) - \sum_{i=1}^{c_\t} \left(\Pr(\mathcal{E}_i)\right)^2$ of $\zeta$, and the variance $\lambda =\sum_{i=1}^{c_\t} \Pr(\mathcal{E}_i)$ of the Poisson approximation becomes large. To tackle the problem, we define a Translated Poisson~\cite{rollin2007translated} random variable
$Y = \lfloor \lambda_2 \rfloor + \Pi_{\lambda - \lfloor \lambda_2 \rfloor}$, where $\lambda_2 = \lambda - \sigma^2$ and 
$\Pi$ is Poisson distribution with parameter 
$ \lambda - \lfloor \lambda_2\rfloor$. 
In this formula 
$\lambda =\sum_{i=1}^{c_\t} \Pr(\mathcal{E}_i)$ is the expected value of distribution $\zeta$. Thus, the difference between the variance of $Y$ and $\zeta$ can be written as: 
\begin{align}
        &\text{Var(Y)}-\text{Var}(\zeta) = \lambda - \lfloor \lambda_2 \rfloor - \sigma^2 =  \lambda - \sigma^2 - \lfloor \lambda_2 \rfloor, \nonumber \\
        & =  \lambda_2 - ( \lambda_2 - \left \{ \lambda_2 \right \})  = \left \{ \lambda_2 \right \} < 1 , 
    \end{align}
    where $\left \{ \lambda_2 \right \}  = \lambda_2 - \lfloor \lambda_2 \rfloor$ . As can be seen the difference between the variances becomes small in this case. 

Equation~\ref{poissform} for translated Poisson changes~to
\begin{align}\label{translatedpoissform}
 & \Pr[\zeta < k] \approx  \Pr[Y < k] 
  = \Pr \Big[ \Big( \lfloor \lambda_2 \rfloor + \Pi_{\lambda - \lfloor \lambda_2 \rfloor} \Big) < k \Big], \nonumber \\
  & = \Pr[\Pi_{\lambda - \lfloor \lambda_2 \rfloor} < k-\lfloor \lambda_2 \rfloor] = \Pr[\Pi_{\lambda - \lfloor \lambda_2 \rfloor} < k-\lfloor \lambda_2 \rfloor-1], \nonumber \\
  &  + \frac{\lambda - \lfloor \lambda_2 \rfloor}{k-\lfloor \lambda_2 \rfloor-1}\Pr[\Pi_{\lambda - \lfloor \lambda_2 \rfloor} = k-\lfloor \lambda_2 \rfloor-2], 
\end{align}
and the complexity of obtaining $\Pr(X_{\G,\t,\l} \geq k)$ remains the same. 



We will now consider the scenario when $c_{\t}$ is large. In this case, the variance of $\zeta$ will be large. In the following, we show the use of Central Limit Theorem for this case.

\medskip
\noindent
\textbf{Central Limit Theorem.} An important theorem in statistics, Lyapunov's Central Limit Theorem (CLT) ~\cite{Lyapunov-Nouvelle} states that, given a set of random variables (not necessarily i.i.d.), their properly scaled sum converges to a normal distribution under certain conditions. 

If $c_\t$ and hence $\sigma^2$ are large, then by ~\cite{Lyapunov-Nouvelle}, 
$Z = \frac{1}{\sigma} \sum_{i=1}^{c_\t} (\mathcal{E}_i-\mu_i)$ has standard normal distribution, where $\mu_i = \Pr(\mathcal{E}_i)$.
To approximate $\Pr[\zeta \geq k] = \Pr[\sum_{i=1}^{c_\t}\mathcal{E}_i \geq k]$ using CLT we can subtract $\sum_{i=1}^{c_\t} \mu_i$ from the sum of $\mathcal{E}_i$'s and divide by $\sigma$. As a result, we have:
\begin{equation}\label{clt}
  \Pr\left[\sum_{i=1}^{c_\t}\mathcal{E}_i \geq k\right] =  \Pr\left[\frac{1}{\sigma} \sum_{i=1}^{c_\t} (\mathcal{E}_i - \mu_i) \geq \frac{1}{\sigma}\left(k- \sum_{i=1}^{c_\t} \mu_i\right)\right]
\end{equation}

Since $ Z = \frac{1}{\sigma} \sum_{i=1}^{c_\t} (\mathcal{E}_i-\mu_i)$ has standard normal distribution, we can find the maximum value of $k$ such that the right-hand side of Equation~\ref{clt} is at equal or greater than the threshold. Evaluation of each probability can be done in constant time. 
Thus, finding the maximum value of $k$ can be done in $O(c_\t)$ time.

\medskip
\noindent
\looseness=-1
\textbf{Binomial Distribution.}
In many 
networks, edge probabilities  are close to each other and as a result, for each triangle $\t$, $\Pr(\mathcal{E}_i)$'s are also close to each other. In that case, the distribution of support of the triangle $\t$ can be well approximated by Binomial distribution. A random variable $X$ is said to have Binomial distribution with parameters $p$ and $n$, if the probability mass function of $X$ is given by~\cite{papoulis2002probability}:
\begin{equation}
     \Pr[X=k] = \binom{n}{k} p^k(1-p)^{(n- k)}.
\end{equation}

In the above equation, $p$ is success probability, and $n$ is the number of experiments. 
In statistics, the sum of non-identically distributed and independent Bernoulli random variables can be approximated by the Binomial distribution \cite{ehm1991binomial}. 
As discussed in \cite{ehm1991binomial}, the Binomial distribution provides a good approximation, if its variance is close to the variance of $\zeta$.
For the approximation, we set $n=c_{\t}$ and $n\cdot p=\mu$.


We observe that tail probabilities for the Binomial distribution can be calculated inexpensively as these probabilities satisfy the following well-known recursive relation 
\begin{align}\label{binomform}
  \Pr[\zeta = k]  = \frac{(n-k+1)p}{k(1-p)}\Pr[\zeta = k-1]. 
\end{align}
Using Equation \ref{binomform}, and iterating over values of $k$ from $1$ up to $c_\t$, we can evaluate $\Pr[\zeta \geq k]$ in $O(1)$ time, and find the maximum $k$ such that $\Pr(\t) \cdot \Pr[\zeta \geq k] \geq \theta$. Thus, the time complexity of obtaining probabilistic support for a triangle $\t$ in this case is $O(c_\t)$.

\nop{
Given two integers $a$, $b$, and a real number $y \in [0,1]$,  the regularized incomplete beta function $I_y(a,b)$ is defined as follows~\cite{press1993incomplete}:
\begin{equation}\label{regbeta}
    I_y(a,b) = \sum_{i=a}^{a+b-1}\binom{a+b-1}{i}y^i(1-y)^{a+b-1-i},
\end{equation}

We use the following corollary to approximate $\Pr[\zeta \geq l]$ in terms of the regularized incomplete beta function.
\begin{cor}\label{corbeta}
\Pr$[\zeta \geq l] \sim  I_p(l,c_\t-l+1)$,
where $p = (c_\t)^{-1} \sum_{i=1}^{c_\t} \big(\Pr(\mathcal{E}_i) \big)$.
\begin{IEEEproof}
\normalfont Since $\zeta$ has Binomial distribution, thus we have,
\begin{equation}\label{corebetaequ}
   \Pr[\zeta \geq l] = \sum_{j=l}^{c_\t}\Pr[\zeta = l]  \sim \sum_{j=l}^{c_\t} \binom{c_\t}{j}p^j(1-p)^{c_\t-j}
\end{equation}

Let $a = l$, and $b= c_\t -l +1$, comparing Equation~\ref{corebetaequ} and Equation~\ref{regbeta}, the right-hand side of Equation~\ref{corebetaequ} is equal to $I_p(l,c_\t-l+1)$. The corollary follows.
\end{IEEEproof}
\end{cor}

Using the table of incomplete beta functions \cite{pearson1934tables}, each $I_p$ can be obtained in constant time. Thus, to find 
maximum $l$ such that $Pr(\t) \cdot I_p(l,c_\t-l+1) \geq \theta$  we need to perform logarithmic search (in the number of $4$-cliques containing $\t$) over the values of $I_p$. This takes $O(\log c_\t)$ time.
}

\medskip
\noindent
{\bf Summary.}
We compute $\Pr(X_{\G,\t,\l} \geq k)$ using 
the following set of conditions based on four thresholds $A,B,C,D$. 
\begin{enumerate}
    \item If $c_{\t}$ is large ($c_{\t} \geq A)$, the {\em CLT} approximation is used.
    \item If (1) does not hold, then if $c_{\t}$ and $\Pr(\mathcal{E}_i)$'s are small ($c_{\t}$ $< B$ and $\Pr(\mathcal{E}_i)'s < C$), the \textit{Poisson} approximation is used.
    \item If (1) and (2) do not hold, then if $\sum_{i=1}^{c_{\t}} \big(\Pr(\mathcal{E}_i) \big)^2 > 1$, the \textit{Translated Poisson} approximation is used.
    \item If (1), (2), and (3) do not hold, then if the ratio of the variance of $\zeta$ to the variance of the Binomial distribution with $n=c_{\t}$ and $p=\mu/n$ is close to $1$ (e.g. not less than a number $D$), the {\em Binomial} approximation is used. 
    \item Otherwise, \textit{Dynamic Programming} is used.
\end{enumerate}


\looseness=-1
\color{black}{For selecting the thresholds we refer to the literature in statistics. In particular, CLT gives a good approximation if the number (for our problem $c_{\t})$ of random variables in the sum is at least $30$ (\cite{mukhopadhyay2000probability}, p. 547). In fact, we set our threshold $A=200$ to much higher than what is suggested by the literature. 
Also, regarding Poisson distribution, the existence probability (for our problem $\Pr(\mathcal{E}_i$)'s) of the indicator random variables in the sum should be less than $0.25$ (see~\cite{le1960approximation}). So, we set $C = 0.25$. 
We set $B$ to be half of $A$ so that it is considerably far from $A$ (threshold on $c_{\t}$). 
%
We set $D=0.9$ which is close enough to 1.  
%
}

\looseness=-1
When using $A=200$, $B=100$, $C=0.25$, $D=0.9$, we observed that the results of computing $\Pr(X_{\G,\t,\l} \geq k)$ 
using an approximation are practically indistinguishable from the solution of dynamic programming. 
Furthermore, as we observed in our experiments, falling back to dynamic programming in point~(5) happens only in a small amount of cases, i.e. most triangles in real world networks satisfy one of the earlier conditions (1)-(4). This means we can avoid dynamic programming for most of the triangles. 

\section{Global and Weakly-Global Nucleus Decomposition}
\label{globalnucleus}

\looseness=-1
In this section, we propose algorithms for computing global and weakly-global nucleus decomposition. 
Given a graph $\H$, computing $\Pr(X_{\H,\t,\g} \geq k)$ and $\Pr(X_{\H,\t,\w} \geq k)$ requires finding all the possible worlds of $\H$, which  are in total $2^{|E(\H)|}$, where $E(\H)$ is the number of edges in $\H$.
This 
is prohibitive. 
Thus, we use Monte Carlo sampling to estimate the probabilities, denoted by $\hat{\Pr}(X_{\H,\t,\g} \geq k)$ and $\hat{\Pr}(X_{\H,\t,\w} \geq k)$. The following lemma states a special version of the Hoeffding’s inequality~\cite{hoeffding1994probability} that provides the minimum number of samples required to obtain an unbiased estimate.
\begin{lem}\label{Mont}
Let $Y_1, \cdots, Y_{n}$ be independent random variables bounded in the interval $[0,1]$. Also, let $\bar{Y}=\frac{1}{n}\sum_{i=1}^{n}Y_i$. Then, we have that
\begin{equation}
    \Pr\left[|\bar{Y}-\mathbb{E}[\bar{Y}]| \geq \epsilon\right] \leq 2 e^{-2n\epsilon^2}.
\end{equation}
In other words, for any $\epsilon, \delta \in (0,1]$,  $\Pr\left[|\bar{Y}-\mathbb{E}[\bar{Y}]| \geq \epsilon \right] \leq \delta$, if $n\geq \left\lceil \frac{1}{2\epsilon^2}\ln\left(\frac{2}{\delta}\right)\right\rceil$.
\end{lem}

Based on the above, using Monte Carlo sampling, we can obtain an estimate of $\Pr(X_{\H,\t,\g} \geq k)$, and $\Pr(X_{\H,\t,\w} \geq k)$ for any subgraph $\H$ by sampling $n$ possible worlds of $\H$, $\left \{ H_1,\cdots,H_{n} \right \}$, where $n= \left \lceil \frac{1}{2\epsilon^2}\ln\left(\frac{2}{\delta}\right) \right \rceil$, $\epsilon$ is an error bound, and $\delta$ is a probability guarantee. In particular, we have:
\begin{equation}\label{estimglob}
  \hat{\Pr}(X_{\H,\t,\mu} \geq k) = \sum_{i=1}^{n} \mathbbm{1}_{\mu}(H_i,\t,k)/n,   
\end{equation}
where $\mu = \g$ or $\w$, and the indicator function $\mathbbm{1}_{\mu}(H_i,\t,k)$ is given in Definition~\ref{lwg}.
Based on Lemma~\ref{Mont}, what we obtain is an unbiased estimate. Thus, setting $\mu = \g,\w$, we have 
\begin{equation}
\Pr\left [ \left | \Pr\left(X_{\H,\t,\mu} \geq k\right) - \hat{\Pr}\left(X_{\H,\t,\mu} \geq k \right) \right | \geq \epsilon \right ] \leq \delta.
\end{equation}

\smallskip
\noindent
$\g$-$\pmb{(k,\theta)}$\textbf{-nucleus.}
In what follows, we propose an algorithm for finding all $\g$-$(k,\theta)$-nuclei for different values of $k = 1, \ldots , k_{\max}$, where $k_{\max}$ is the largest value for which the local nucleus is non-empty. 
This is because we extract global nuclei from local ones since every $\g$-$(k,\theta)$-nucleus 
is part of an $\l$-$(k,\theta)$-nucleus.
The main steps of our proposed algorithm are given in Algorithm~\ref{global1}. 

Given a positive integer $k$, threshold $\theta$, error-bound $\epsilon$, and confidence level $\delta$, the algorithm starts by creating subgraph $\C_{\color{black}{k}}$ as the union of all $\l$-$(k,\theta)$-nuclei (line~\ref{originalcand}).  
Then, the algorithm incrementally builds a candidate 
$\g$-$(k,\theta)$-nucleus $\H$ as follows. 
For each triangle $\t$ in $\C_{\color{black}{k}}$, it adds to $\H$ all the 4-cliques in $\C_{\color{black}{k}}$ containing $\t$ (line~\ref{createfirst}). By this process other triangles $\t'$ could potentially be added to $\H$ such that the number of 4-cliques containing $\t'$ is less than $k$. In order to remedy this, the algorithm adds all the 4-cliques of $\C_{\color{black}{k}}$ containing $\t'$ to $\H$.   
This process continues until all the triangles in $\H$ are contained in at least $k$ $4$-cliques  (lines~\ref{whileglob1}-\ref{whileglob2}). 
Once the candidate graph $\H$ is obtained, $n$ samples of possible worlds of $\H$ are obtained (line~\ref{samplepossible}). Then, the algorithm checks if the condition $\hat{\Pr}(X_{\H,\t,\g} \geq k) \geq \theta$ is satisfied for each triangle $\t$ in $\H$ (lines~\ref{ifglob1}-\ref{ifglob2}). At the end, the algorithm returns all $\g$-$(k,\theta)$-nuclei 
$\H$ (line~\ref{globcond1}-\ref{globcond2}), \color{black}{for all the possible values of $k$}.

\begin{algorithm}
\caption{\color{black}{$\g$}-{\sl NuDecomp}}
\label{global1}
\begin{algorithmic}[1] 
\Function{\color{black}$\g$\_Nucleus}{$\color{black}\G,$ $\color{black}\theta$, $\color{black}\epsilon$, $\color{black}\delta$}
\State{\text{solution} $\gets \left \{  \right \}$}
\ForAll{\color{black}{$\ k \gets 1$ \textbf{to} $k_{\max}$}}
\State{$\C_{\color{black}{k}} \gets$ the union of $\l$-$(k,\theta)$-nuclei by Algorithm~\ref{peelingnuclei}}\label{originalcand}
\ForAll{$\t \in \C_{\color{black}{k}}$}
\State{$\H\gets$ all 4-cliques in $\C_{\color{black}{k}}$ containing $\t$} \label{createfirst}
\While {$\exists \t' \in \H$ with less than $k$ 4-cliques $ \in \H$ \par \hskip\algorithmicindent containing it} \label{whileglob1}
\State{add all $4$-cliques of $\C_{\color{black}{k}}$ containing $\t'$ to $\H$} \label{whileglob2}
\EndWhile
\State{\textit{condition\_hold} $\gets \textbf{true}$}
\State{$\textit{sample} \gets \left \{ H_1,\cdots,H_{n} \right \}$}\label{samplepossible}
\ForAll{$\t \in \H$} $\hat{\Pr}(X_{\H,\t,\g} \geq k) \gets \text{Eq}.(\ref{estimglob})$ \label{ifglob1}
\If{$\hat{\Pr}(X_{\H,\t,\g} \geq k) < \theta$} 
\State{\textit{condition\_hold} $\gets \textbf{false}$} \label{ifglob2}
\State{\textbf{break}}
\EndIf
\EndFor
\If{\textit{condition\_hold} $ == \textbf{true}$}\label{globcond1}
\State{$ \text{solution} \gets \text{solution} \cup \H$}
\EndIf
\EndFor
\EndFor
\State{\textbf{return} \text{solution}} \label{globcond2}
\EndFunction
\end{algorithmic}
\end{algorithm}
\begin{algorithm}
\caption{\color{black}{$\w$}-{\sl NuDecomp}}
\label{global2}
\begin{algorithmic}[1] 
\Function{\color{black}$\w$\_Nucleus}{$\color{black}\G,$ $\color{black}\theta$, $\color{black}\epsilon$, $\color{black}\delta$}
\State{\text{solution} $\gets \{ \}$}
\ForAll{\color{black}{$\ k \gets 1$ \textbf{to} $k_{\max}$}}
\ForAll{$\l$-$(k,\theta)$ $\H$}
\State{\textit{global\_score}[$\t$] $\gets 0$ \text{for each $\t \in \H$}} \label{globstp1}
\State{$\textit{sample} \gets \left \{ H_1,\cdots , H_{n} \right \}$} \label{globstp2}
\ForAll{$H \in \textit{sample}$} \label{globstp3}
\State{ $H' \gets$ $k$-nucleus of $H$} \label{globstp4}
\ForAll{triangle $\t \in H'$} \label{globstp5}
\State{\textit{global\_score}[$\t$] $++$  
} 
\label{globstp6}
\EndFor
\EndFor
\ForAll{$\t \in \H$} \label{globstp7}
\State{$\hat{\Pr}(X_{{H},\t,\w} \geq k) \gets \textit{global\_score}[\t]/n$ } \label{globstp8}
\EndFor
\State{solution $ \gets \text{solution } \cup$ connected union of $\t$'s \par \hskip\algorithmicindent with $\hat{\Pr}(X_{\H,\t,\w} \geq k) \geq \theta$} \label{globstp9}
\EndFor
\EndFor
\State{\textbf{return} solution} \label{globstp10}
\EndFunction
\end{algorithmic}
\end{algorithm}

\noindent
\textbf{$\w$-$(k,\theta)$-nucleus.} In what follows, we propose an algorithm for finding all $w$-$(k,\theta)$-nuclei, \color{black}{for different values of $k = 1, \ldots , k_{\max}$, where $k_{\max}$ is as before.} 
We begin by noting that each $\w$-$(k,\theta)$-nucleus is an $\l$-$(k,\theta)$-nucleus. 
The steps of our proposed algorithm are given in Algorithm~\ref{global2}. 

\smallskip
\noindent
We use array \textit{global\_score} to store the number of deterministic $k$-nuclei that each triangle belongs to. The array is initialized to zero for all the triangles in the candidate graph (line~\ref{globstp1}). 
For each candidate graph which is a $\l$-$(k,\theta)$-nucleus, we obtain 
the required number $n$ of possible worlds for the given $\epsilon$ and $\delta$.
Then, we perform deterministic nucleus decomposition on each world (lines~\ref{globstp2}-\ref{globstp4}). 
If triangle $\t$ is in a deterministic $k$-nucleus of that possible world, the corresponding index of $\t$ in array \textit{global\_score} is incremented by one (lines~\ref{globstp5}-\ref{globstp6}). In line~\ref{globstp8}, the approximate value $\hat{\Pr}(X_{\H,\t,\w} \geq k)$ is obtained for each triangle. Then, we start creating the connected components $\H$ using triangles with $\hat{\Pr}(X_{\H,\t,\w} \geq k) \geq \theta$ (line~\ref{globstp9}). 
At the end, the algorithm returns all $\w$-$(k,\theta)$-nuclei, \color{black}{for all the possible values of $k$}.

\looseness=-1
\noindent
{\bf Remark.} Both of these algorithms run in polynomial time. They compute the correct answer provided the estimation of the threshold probabilities using Monte-Carlo sampling is close to the true value. If not, they give approximate solutions. 

\noindent
\looseness=-1
\color{black}{\textbf{Space Complexity.} For global and weakly global decompositions the space needed is $O(T_{\mathcal{G}} + m\cdot n)$, where $m$ is the number of edges in $\H$ and $n$ is the number of possible worlds for $\H$ we sample.}

\color{black}{From a theoretical point of view, $n$, the number of samples, is constant for fixed values of $\epsilon$ and $\delta$, and since $m$, number of edges, is absorbed by $T_{\G}$, we can say that the above complexity is again $O(T_{\G})$, i.e. same as the space complexity for deterministic nucleus.}

\looseness=-1
\color{black}{From a practical point of view, for each sample graph (possible world) we use a bit array to record whether an edge exists in the sample or not. 
For practical values of $\epsilon$ and $\delta$, $m\cdot n$ is about $200\cdot m$ bits, which is $200/(32+32) \sim 3$ times more than the space needed to store the edges as adjacency lists (assuming an integer node id is 32 bits, and the graph is undirected, i.e. each edge is represented as two directed edges). 
In other words, to store the $n$ possible worlds we only need about three times more space 
than what is needed to store $\G$.}

\section{Experiments}
\label{experiments}

We present our extensive experimental results to test the efficiency, effectiveness, and accuracy of our proposed algorithms. Our implementations are in Java and the experiments are conducted on a commodity machine with Intel i7, 2.2Ghz CPU, and 12Gb RAM, running Ubuntu 18.04.

\smallskip
\looseness=-1
\noindent
{\bf Datasets and Experimental Framework.}
Statistics for our datasets are in Table~\ref{Dataset Statistics}. 
We order the datasets based on the number of triangles they contain. 
Datasets with real probabilities are \textit{flickr}, \textit{dblp}, and \textit{biomine} from  \cite{bonchi2014core,potamias2010k}
and 
\textit{krogan} from \cite{krogan2006global}.




\looseness=-1
We also consider datasets \textit{ljournal-2008} and \textit{pokec}. 
\textit{ljournal-2008} is obtained from 
Laboratory of Web Algorithmics (\url{http://law.di.unimi.it/datasets.php}) and \textit{pokec} 
is from the Stanford Network Analysis Project (\url{http://snap.stanford.edu}). 
For these networks, we generated edge probabilities uniformly distributed in $(0,1]$.

\begin{table}[t]
{\small
\centering
\begin{tabular}{lrrrrr} \hline 
\small Graph &           \small $|V|$ & \small $|E|$ & $\text{d}_{\max}$ & \small $p_{avg}$ & \small $ \left | \t \right |$  \\ \hline \hline
\small krogan 		& \small 2,708		& \small 7,123 & \small 141 & \small 0.68 & \small 6,968\\
\small dblp			& \small 684,911	& \small 2,284,991 & \small 611 & \small 0.26 & \small 4,582,169\\
\small flickr		& \small 24,125		& \small 300,836 & \small 546 & \small 0.13 & \small 8,857,038\\
\small pokec & \small 1,632,803 & \small 22,301,964 & \small 14,854 & \small 0.50 & \small 32,557,458 \\
\small biomine		& \small 1,008,201	& \small 6,722,503 & \small 139,624 & 0.27 & \small 93,716,868\\ 
\small ljournal-2008	& \small 5,363,260	& \small 49,514,271 & \small 19,432 & 0.50 & \small 411,155,444 \\ \hline
\end{tabular}
\caption{Dataset Statistics}
\label{Dataset Statistics}
\vspace{-0.5cm}
}
\end{table}

\looseness=-1
We evaluate our algorithms on three important aspects. First is the efficiency. For this, we report the running time of our algorithms in Subsection \ref{efficiency}. Second is the accuracy and closeness of our approximation methods. We discuss this in Subsection \ref{accuracy}. Third is the quality of the output nucleus as measured by \textit{density} and \textit{probabilistic clustering coefficient} which are discussed in Subsection~\ref{density}. \color{black}{Finally, in Subsection~\ref{casestudy} we show the usefulness of our nucleus definitions over probabilistic graphs by presenting a detailed use-case.} 
\begin{figure*}
    \centering
     \subfloat{ 
    \begin{tikzpicture}
\begin{axis}[width=3.3cm,height=3cm,
xtick pos=left,
ytick pos=left,
    title={\textbf{krogan}},
    xlabel style={font=\fontsize{7}{7}\selectfont},
  symbolic x coords={0.1,0.2,0.3,0.4,0.5},
    ytick distance =0.03,
    xtick={0.1,0.2,0.3,0.4,0.5},
     ylabel style={font=\fontsize{1}{1}\selectfont},
    ylabel={\textbf{Time (s)}},
    ticklabel style = {font=\fontsize{7}{6}\selectfont},
    ylabel near ticks,
    xticklabel style={rotate=-45},
]

\addplot[
    color=blue,
    mark=square,
    ]
    coordinates {
    (0.1,0.47)(0.2,0.44)(0.3,0.43)(0.4,0.38)(0.5,0.34)
    };
    \addplot[
    color=red,
    mark=*,
    ]
    coordinates {
    (0.1,0.47)(0.2,0.42)(0.3,0.38)(0.4,0.37)(0.5,0.34)
    };
    
\end{axis}
\end{tikzpicture}
    }
    \subfloat{ 
    \begin{tikzpicture}
\begin{axis}[width=3.3cm,height=3cm,
xtick pos=left,
ytick pos=left,
    title={\textbf{dblp}},
     xlabel style={font=\fontsize{7}{7}\selectfont},
    symbolic x coords={0.1,0.2,0.3,0.4,0.5},
    xticklabel style={rotate=-45},
    xtick={0.1,0.2,0.3,0.4,0.5},
    ylabel near ticks,
    ticklabel style = {font=\fontsize{7}{6}\selectfont},
    legend style={draw=none,nodes={scale=0.7}},
]

\addplot[
    color=blue,
    mark=square,
    ]
    coordinates {
    (0.1,25)(0.2,13.08)(0.3,10.77)(0.4,10.1)(0.5,9.8)
    };
    \addplot[
    color=red,
    mark=*,
    ]
    coordinates {
    (0.1,17)(0.2,10.73)(0.3,9.65)(0.4,9.45)(0.5,9.1)
    };
    
\end{axis}
\end{tikzpicture}
    }
    \subfloat{ 
    \begin{tikzpicture}
\begin{axis}[width=3.3cm,height=3cm,
xtick pos=left,
ytick pos=left,
    title={\textbf{flickr}},
     xlabel style={font=\fontsize{7}{7}\selectfont},
   symbolic x coords={0.1,0.2,0.3,0.4,0.5},
    xticklabel style={rotate=-45},
    ytick distance =3,
    xtick={0.1,0.2,0.3,0.4,0.5},
    ylabel near ticks,
    ticklabel style = {font=\fontsize{7}{6}\selectfont},
    legend style={draw=none,nodes={scale=0.7}},
]

\addplot[
    color=blue,
    mark=square,
    ]
    coordinates {
    (0.1,16.036)(0.2,12.437)(0.3,10.932)(0.4,9.775)(0.5,9.182)
    };
    \addplot[
    color=red,
    mark=*,
    ]
    coordinates {
    (0.1,13.291)(0.2,10.464)(0.3,10.063)(0.4,8.983)(0.5,8.902)
    };
    
\end{axis}
\end{tikzpicture}
    }
\subfloat{ 
    \begin{tikzpicture}
\begin{axis}[width=3.3cm,height=3cm,
xtick pos=left,
ytick pos=left,
   title={\textbf{pokec}},
     xlabel style={font=\fontsize{7}{7}\selectfont},
   ymode=log,
    symbolic x coords={0.1,0.2,0.3,0.4,0.5},
    xticklabel style={rotate=-45},
    ticklabel style = {font=\fontsize{7}{6}\selectfont},
    xtick={0.1,0.2,0.3,0.4,0.5},
    ylabel near ticks,
    legend style={draw=none,nodes={scale=0.7}},
]

\addplot[
    color=blue,
    mark=square,
    ]
    coordinates {
    (0.1,672)(0.2,400)(0.3,298)(0.4,261)(0.5,238)
    };
    \addplot[
    color=red,
    mark=*,
    ]
    coordinates {
   (0.1,474)(0.2,336)(0.3,271)(0.4,230)(0.5,226)
    };
    
\end{axis}
\end{tikzpicture}
   }
\subfloat{ 
    \begin{tikzpicture}
\begin{axis}[width=3.3cm,height=3cm,
xtick pos=left,
ytick pos=left,
   title={\textbf{biomine}},
    ymode=log,
     xlabel style={font=\fontsize{7}{7}\selectfont},
  symbolic x coords={0.1,0.2,0.3,0.4,0.5},
    xticklabel style={rotate=-45},
    ticklabel style = {font=\fontsize{7}{6}\selectfont},
    xtick={0.1,0.2,0.3,0.4,0.5},
    ylabel near ticks,
    legend style={draw=none,nodes={scale=0.7}},
]

\addplot[
    color=blue,
    mark=square,
    ]
    coordinates {
    (0.1,1098.03)(0.2,1005.31)(0.3,939.734)(0.4,878.526)(0.5,770.46)
    };
    \addplot[
    color=red,
    mark=*,
    ]
    coordinates {
   (0.1,838.548)(0.2,771)(0.3,776)(0.4,747.957)(0.5,753)
    };
    
\end{axis}
\end{tikzpicture}
   }
    \subfloat{ 
    \begin{tikzpicture}
\begin{axis}[width=3.3cm,height=3cm, declare function={myinf=10^5;},ymode=log,
ymin=1000,ymax=10^5,
    yticklabel={
    $10^{\the\numexpr\ticknum+2}$
    },
legend pos=outer north east,
xtick pos=left,
ytick pos=left,
   title={\textbf{ljournal}},
  xlabel style={font=\fontsize{7}{7}\selectfont},
    xticklabel style={rotate=-45},
     xtick={0.1,0.2,0.3,0.4,0.5},
    ticklabel style = {font=\fontsize{7}{6}\selectfont},
    ylabel near ticks,
   legend entries={\textbf{DP},\textbf{AP}},
    legend style={font=\fontsize{7.5}{6}\selectfont,draw=none,nodes={scale=0.6}},
]

\addplot[
    color=blue,
    mark=square,
    ]
    coordinates {
    (0.2,36820) (0.3,9322)(0.4,3064)(0.5,1624)
    };
    \addplot[
    color=red,
    mark=*,
    ]
    coordinates {
   
  (0.1,29786) (0.2,7214) (0.3,3322) (0.4,1864)(0.5,1390)
    };
    
\end{axis}
\end{tikzpicture}
   }
   \vspace{-0.1cm}
    \caption{\textbf{Run times of DP and AP for varying $\theta$ (x axis). Both perform well on medium datasets. For bigger datasets, biomine and ljournal, the difference is more pronounced. For ljournal, for $\theta=0.1$, it is only AP that can complete within one day. }}
    \label{runningtime}
    \vspace{-0.2cm}
\end{figure*}
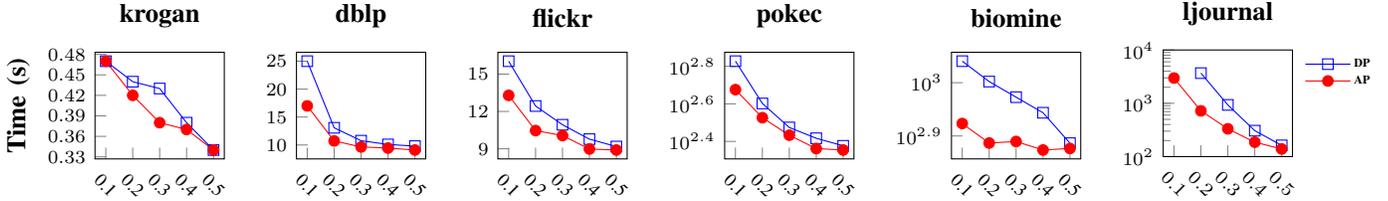

\subsection{Efficiency Evaluation}\label{efficiency}
In this section, we report the running times of our proposed algorithms for local nucleus decomposition: one that uses dynamic programming and the other that uses statistical approximations for computing and updating the support of triangles. 
We denote them by DP and AP, respectively.
Next, we report the running times of our (fully) global and weakly-global nucleus decomposition algorithms, which we denote by FG and WG. We set error-bound $\epsilon = 0.1$ and confidence level $\delta = 0.1$. \color{black}{Based on these values and Lemma~\ref{Mont}, we set the number of samples to $n=200 > \left \lceil \frac{1}{2\epsilon^2}\ln\left(\frac{2}{\delta}\right) \right \rceil$ (i.e. greater than what is required by Hoeffding's inequality). As such, our results for global and weakly-global notions of nuclueus are approximate but come with strong quality guarantees. }

\looseness=-1
Running time results for DP and AP are shown in Figure~\ref{runningtime} for different values of $\theta$. Y-axis for the last 3 plots is in log-scale. 

\looseness=-1
Both algorithms perform well on medium-size datasets, \textit{dblp} and \textit{flickr}; computing the nucleus decomposition of these two graphs takes less than $1$ sec. For a larger-size dataset, \textit{pokec}, both algorithms complete in less than $10$ min. 
Note that AP clearly outperforms DP on large-size datasets such as \textit{biomine} and \textit{ljournal-2008} for small values of $\theta$. 
For instance, for \textit{ljournal-2008} with $\theta=0.1$, 
it is only AP that can run to completion, whereas 
DP could not complete after one day. 
Nevertheless, both DP and AP are able to run in reasonable time for all the other cases, which is good considering that nucleus decomposition is a harder problem than core and truss decomposition.

In general, the running times of both DP and AP decrease significantly as $\theta$ increases. 
This is because the number of triangles which, 
(a) exist with probability greater than $\theta$ and 
(b) have a support at least $k$ again with probability greater than $\theta$, decreases. 
As can be seen, AP is faster than DP on all datasets for different values of $\theta$. 
In addition to the \textit{ljournal-2008} case for which only AP could complete, when $\theta=0.1$, the gain of AP over DP is about $24\%$ and $30\%$ for \textit{biomine} and \textit{pokec}, respectively.

{\color{black}
For speed-up evaluation of AP vs. DP we added two more datasets. The statistics of these datasets are given in Table~\ref{Dataset Statistics respon}.
The first dataset is \textit{enwiki-2013}. 
What is special about this dataset is that its maximum initial nucleus score is $2,813$, which is much larger than in other graphs we consider. 
We set $\theta=0.1$; 
when $\theta$ is small, more triangles can have enough probability to be part of a much larger number of $4$-cliques.  This can cause too much work for DP to compute nucleus scores and update these values when the triangles are being processed in the peeling step. 
For this dataset, DP was not able to complete the computation within one week. In contrast, AP completed in about 80K sec (less than a day). 

The other additional dataset we considered is \textit{itwiki-2013}. The maximum initial nucleus score in this dataset is $1,866$. 
In this graph, using the same $\theta=0.1$, DP needs about 40h, whereas AP $16h$, i.e. AP is 2.5 times faster than DP. 

Moreover, we ran DP and AP on  \textit{biomine} with $\theta = 0.01$. DP took about 37.5h, whereas AP 2.5h, thus being 15 times faster than DP.
}

\begin{table}[hp]
 \color{black}
{\small
\centering
\begin{tabular}{lrrrrr} \hline 
\small Graph &           \small $|V|$ & \small $|E|$ & $\text{d}_{\max}$ & \small $p_{avg}$ & \small $ \left | \t \right |$  \\ \hline \hline
\small enwiki-2013 		& \small 4,206,785		& \small 91,939,728 & \small 432,260 & \small 0.5 & \small 304,083,160\\
\small itwiki-2013 & \small 1,016,867 & \small 23,429,644 & \small 91,517 & \small 0.5 & 89,901,299 \\
\hline
\end{tabular}
\caption{\color{black}{Additional datasets. $\pmb{|V|}$, $\pmb{|E|}$, $\pmb{d_{\max}}$, $\pmb{p_{avg}}$, $\pmb{\bigtriangleup}$, are number of vertices, edges, maximum degree, average edge probability, and number of triangles in the graph, respectively.} }
\label{Dataset Statistics respon}
}
\end{table}

We report the running time of FG and WG in Figure~\ref{runtglob} \color{black}{along with the running time of local (denoted by $L$ in the figure) nucleus decomposition for $\theta =0.1$ (which as explained above is more difficult than $\theta =0.2, \ldots, 0.5$). Note that the global and weakly-global nuclei are obtained from the local ones using Algorithms~\ref{global1} and ~\ref{global2}. Therefore, their running time includes the time required for obtaining local nuclei. For local decomposition, we use DP to obtain the probabilistic support of the triangles, except for \textit{ljournal-2008} for which we use AP since DP does not scale for this threshold. Also, we report running times averaged across 5 runs, since the solutions of FG and WG depend on the random sampling steps. 
} 

\looseness=-1
In general WG is faster than FG. This is because WG performs deterministic nucleus decomposition only on a fixed number of sample graphs while FG does the decomposition every time that a candidate graph is detected. \color{black}{We also note that as the graph becomes larger, WG will have to perform nucleus decomposition on larger sample graphs leading to increased running time. For FG, usually candidate graphs are small even for large graphs. So, when the graph becomes lager, the runtime of WG increases more compared to FG.}

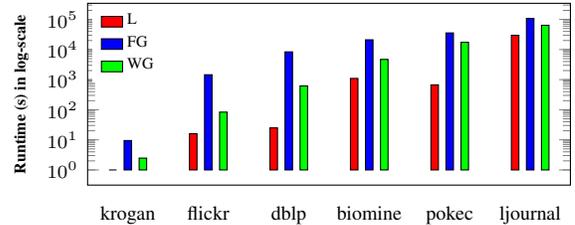
\begin{figure}[htbp]
    \centering
    \begin{tikzpicture}
\begin{axis}[
ybar=0.1cm, 
width=8cm,
height=4cm,
bar width = 0.1cm,
legend style={font=\fontsize{7.5}{6}\selectfont,
draw = none, 
at={(0.01,0.999)},anchor=north west,nodes={scale=0.8}}, legend cell align={left},
legend entries={L,FG,WG},
ymin = 0,
ytick align=inside,
ytick distance =10,
xtick style={draw=none},
ymode=log,
ylabel={\textbf{Runtime (s) in log-scale}},
ylabel near ticks,
ylabel style={font=\fontsize{6}{6}\selectfont},
symbolic x coords={krogan, flickr,dblp,biomine,pokec,ljournal},
xtick=data,
xlabel style = {font=\fontsize{6.5}{6}\selectfont},
x tick label style={/pgf/number format/1000 sep=},
ticklabel style = {font=\fontsize{7.5}{6}\selectfont},
every node near coord/.append style={align=left, rotate=60, anchor=west,font=\fontsize{6}{6}\selectfont},
point meta=rawy,
]
\addplot[fill = red] coordinates {
(krogan,1) (flickr, 16) (dblp, 25) (biomine,1098) (pokec,672) (ljournal,29786)
%
};
\addplot[fill = blue] coordinates {
(krogan,9.47) (flickr, 1454.424) (dblp, 8369.635) (biomine,21098) (pokec,35672) (ljournal,107786)
%
};
\addplot[fill = green] coordinates {
(krogan,2.47) (flickr, 83.364) (dblp,619.046) (biomine, 4708.541) (pokec, 17373.848) (ljournal,63927)
};
\end{axis}
\end{tikzpicture}
\vspace{-0.2cm}
\caption{
Run time of L, FG, and WG. FG and WG include the time for L. WG is faster because it performs deterministic decomposition only on a fixed number of sample graphs while FG does so each time a candidate graph is discovered.}
\label{runtglob}
\vspace{-0.2cm}
\end{figure}

{
\color{black}
Moreover, we compare the running time of nucleus decomposition algorithms, local, weakly-global, and global, on \textit{biomine} with $\theta=0.1$ and $\theta=0.01$ in Figure~\ref{runtimebiomine}. 
For the local decomposition (L) we used DP because we are interested in the relative difference in running time for the different nucleus notions and L is the initial step for computing WG and FG. 

\looseness=-1
When $\theta$ decreases, running times increase since more triangles can have enough probability to be contained in a local nucleus subgraph. 
In terms of the size of the results, Table~\ref{biominenucleussize} shows the average number of vertices and edges for the L, WG, and FG subgraphs aggregated over all $k \in [1,k_{\max}]$.
In general, the average values increase as we decrease threshold. This is due to the fact that by decreasing $\theta$ more triangles can have enough probability to be contained in $4$-cliques. 
}
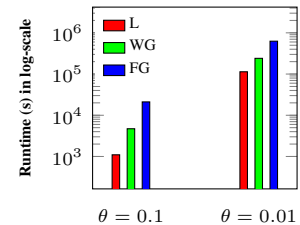
\begin{figure}[htbp]
    \centering
    \begin{tikzpicture}
\begin{axis}[
ybar=0.1cm, 
width=4.3cm,
height=4cm,
bar width = 0.1cm,
legend style={font=\fontsize{7.5}{6}\selectfont,
draw = none, 
at={(0.01,0.999)},anchor=north west,nodes={scale=0.8}}, legend cell align={left},
legend entries={L,WG, FG},
enlargelimits=0.3,
enlarge y limits  = 0.3,
ymin = 0,
ytick align=inside,
ytick distance =10,
xtick style={draw=none},
ymode=log,
ylabel={\textbf{Runtime (s) in log-scale}},
ylabel near ticks,
ylabel style={font=\fontsize{6}{6}\selectfont},
symbolic x coords={$\theta=0.1$, $\theta=0.01$},
xtick=data,
xlabel style = {font=\fontsize{6.5}{6}\selectfont},
x tick label style={/pgf/number format/1000 sep=},
ticklabel style = {font=\fontsize{7.5}{6}\selectfont},
every node near coord/.append style={align=left, rotate=60, anchor=west,font=\fontsize{6}{6}\selectfont},
point meta=rawy,
]
\addplot[fill = red] coordinates {
($\theta=0.1$,1098) ($\theta=0.01$, 113546)
%
};
\addplot[fill = green] coordinates {
($\theta=0.1$,4708) ($\theta=0.01$, 240657)
%
};

\addplot[fill = blue] coordinates {
($\theta=0.1$,21098) ($\theta=0.01$, 627092)
%
};

\end{axis}
\end{tikzpicture}
\vspace{-0.2cm}
\caption{\color{black}Running time (in sec) of local (L), weakly-global (WG), and (fully) global (FG) decomposition on \textit{biomine} for $\theta=0.1$ and $\theta=0.01$.}
\label{runtimebiomine}
\end{figure}

\begin{table}
\color{black}
  \centering
    \begin{tabular}{lrr}
    \hline
          & \multicolumn{1}{c}{$\theta = 0.1$} &  \\
    \midrule
    Model & \multicolumn{1}{l}{Number of Vertices} & \multicolumn{1}{l}{Number of edges} \\
    \midrule
    Local & 75    & 3455 \\
    Weakly-Global & 15      & 157 \\
    Global &   4    &  6\\
    \hline
          & \multicolumn{1}{c}{$\theta=0.01$} &  \\
    \midrule
    \midrule
    Local & 93   & 3785 \\
    Weakly-Global &  55     & 2332 \\
    Global & 5   & 10 \\
    \bottomrule
    \end{tabular}%
    \caption{\color{black}Average number of vertices and edges for local, weakly-global, and global nucleus subgraphs with $\theta$ equal to $0.1$ and $0.01$.}
  \label{biominenucleussize}%
\end{table}%

\begin{table}
  \centering
  \small
    \begin{tabular}{lll}
    \toprule
    \multicolumn{1}{c}{\multirow{2}[4]{*}{Dataset}} & \multicolumn{1}{c}{Avg Error} & \% of  $\t$ with Error \\
\cmidrule{2-3}          & \multicolumn{2}{c}{$\theta=0.2$/$\theta=0.4$} \\
    \midrule
    krogan & 0.0524/0.0209 & 5.24/2.08 \\
    dblp  & 0.0069/0.0041 & 0.69/ 0.41 \\
    flickr & 0.0031/0.0 & 0.31/0.0 \\
    pokec & 0.0014/4.15e-5 & 0.14/0.004 \\
    biomine & 0.0/0.0 & 0.0/0.0 \\
    ljournal-2008 & 0.0179/0.0070 & 1.79/0.69 \\
    \bottomrule
    \end{tabular}%
    \caption{\small Avg difference (error) of AP scores from true DP scores and pct's of triangles with error. Errors are very small.}
    \label{accuracytable}
    \vspace{-0.3cm}
\end{table}%

\subsection{Accuracy Evaluation}\label{accuracy}
To evaluate the accuracy of the AP algorithm, we compare the final nucleus scores obtained by DP and AP algorithms. We report the results in Table~\ref{accuracytable}. We show the results for $\theta$ equal to $0.2$ and $0.4$, since for the remaining values the error results do not differ significantly. The second column shows the average difference (error) from true value over the total number of triangles. The last column shows the percentage of triangles whose value is different from their exact value. 

As can be seen, the average error is quite small for all the datasets we consider. Particularly, for \textit{flickr} with $\theta=0.4$ and \textit{biomine} with $\theta=0.2$ and $\theta=0.4$ we have that AP computes nucleus decomposition with \textit{zero error}. Also, the percentage of triangles with an error score is very small, namely less than $1\%$ for all the datasets, except {\em krogan} and {\em ljournal-2008}. For these two, the percentages are still small, 5.24\% and 1.79\%, respectively. 
These results show that the output of AP is very close to that of exact computation by DP, and thus, AP is a reliable approximation methodology.

\begin{table*}[htbp]
{\small
  \centering
    \begin{tabular}{l c c c c c c c }
    \hline
    \small Graph &   $\theta$    &  $\left | V_{N} \right | / \left |  V_T \right |/\left | V_C \right |$    &   $\left |E_{N}  \right |/\left | E_T \right |  /\left | E_C \right | $    &   $k_{Nmax}/k_{Tmax}/k_{Cmax}$    &   $ \text{PD}_{N}/\text{PD}_T/\text{PD}_C$ &   $\text{PCC}_{N}/\text{PCC}_T/\text{PCC}_C$ & $ \color{black} \text{Time}_N/\text{Time}_T/\text{Time}_C$ \\
    \midrule
   \small dblp  &    0.1   &    $19/34/115$   &    $171/561/6555$  &   $9/14/26$    & $0.800/0.611/0.264$   & $0.790/0.620/0.317$ & $  \color{black} 25/100/15.86$  \\
      \small dblp  &    0.3   &    $14/26/138$   &    $108/366/6693$  & $7/11/23$  & $0.9917/0.785/0.277$ & $0.9918/0.789/0.384$ & $  \color{black} 11/30/16.99$  \\
    \midrule
       \small pokec &    0.1   &  $13/72/288$     &   $121/1335/10592$    &   $3/8/27$    &    $0.678/0.341/0.129$   & $ 0.636/0.393/0.170$ & $ \color{black} 672/1162/4401$\\
       \small pokec &    0.3   &  $6/71/278$     &   $21/1031/10142$    &   $2/6/25$    &    $0.815/0.321/0.132$   & $0.793/0.406/0.172$ & $ \color{black} 298/980/4349$\\
    \midrule
       \small biomine &  0.1  &   $103/102/430$    &  $5231/5127/92200$     &   $18/33 /79$    &   $0.540/0.538 /0.211 $    & $0.540/0.538/0.217$ & $ \color{black} 1098/7642/5792$ \\
       \small biomine &  0.3   &    $7/102 /431 $   &   $23/5125/92625$    &   $2/28/73$    &    $0.714/0.538/0.212 $   & $0.701/0.539/ 0.218$ & $ \color{black} 939/1563/5685$\\
    \hline
    \end{tabular}%
    \caption{\textbf{Cohesiveness statistics of $\pmb{l}$-$\pmb{(k,\theta)}$-nucleus N, $\pmb{(k,\theta)}$-truss, T, and $\pmb{(k,\theta)}$-core, C on \textit{dblp}, \textit{pokec}, and \textit{biomine}. The number of vertices ($\pmb{\left | V_{N} \right | / \left |  V_T \right |/\left | V_C \right |}$), the number of edges ($\pmb{\left |E_{N}  \right |/\left | E_T \right |  /\left | E_C \right | }$), maximum nucleus/truss/core score ($\pmb{k_{Nmax}/k_{Tmax}/k_{Cmax}}$), the probabilistic density (PD$\pmb{_{N}/}$PD$\pmb{_{T}/}$PD$\pmb{_{C}/}$), and the probabilistic clustering coefficient (PCC$\pmb{_{N}/}$PCC$\pmb{_{T}/}$PCC$\pmb{_{C}/}$), respectively. 
    }
    }
  \label{densitycompare}%
  \vspace{-0.4cm}
  }
\end{table*}%
\begin{figure}
\centering
\subfloat{ 
\begin{tikzpicture}
\begin{axis}[
ybar,
width=4.5cm,
height=4cm,
bar width = 0.1cm,
legend style={font=\fontsize{7.5}{6}\selectfont,
draw = none, 
at={(0.01,0.999)},anchor=north west,nodes={scale=0.8}}, legend cell align={left},
legend entries={},
enlargelimits=0.15,
ymax=1,
ymin = 0.7,
ytick align=inside,
ytick distance =0.05,
xtick style={draw=none},
ylabel={\textbf{Average PD/PCC}},
ylabel near ticks,
ylabel style={font=\fontsize{6}{6}\selectfont},
symbolic x coords={1,5,10,%
15},
xtick=data,
xlabel={\textbf{k}},
xlabel style = {font=\fontsize{6.5}{6}\selectfont},
ticklabel style = {font=\fontsize{7.5}{6}\selectfont},
legend style={font=\fontsize{7.5}{6}\selectfont,
draw = none, 
at={(0.999,0.999)},anchor=north east,nodes={scale=0.7}}, legend cell align={left},
legend entries={\textbf{PD},\textbf{PCC}},
]
\addplot[fill = red] coordinates { 
(1,0.80) (5,0.86) (10,0.92) (15,0.93)
};
\addplot[fill = storeClusterComponent] coordinates { 
(1,0.75) (5,0.85) (10,0.89) (15,0.9) 
};

\end{axis}
\end{tikzpicture}
}
\hspace{0.1cm}
\subfloat{ 
\begin{tikzpicture}
\begin{axis}[
ybar,
width=4.5cm,
height=4cm,
bar width = 0.1cm,
legend style={font=\fontsize{7.5}{6}\selectfont,
draw = none, 
at={(0.01,0.999)},anchor=north west,nodes={scale=0.7}}, legend cell align={left},
legend entries={},
enlargelimits=0.15,
enlarge y limits  = 0.004,
ymode = log,
ymax=1500,
ymin = 1,
ytick align=inside,
xtick style={draw=none},
ylabel={\textbf{Average \#of Nuclei/Edges}},
ylabel near ticks,
ylabel style={font=\fontsize{6}{6}\selectfont},
symbolic x coords={1,5,10,%
15},
xtick=data,
xlabel={\textbf{k}},
xlabel style = {font=\fontsize{6.5}{6}\selectfont},
ticklabel style = {font=\fontsize{7.5}{6}\selectfont},
legend style={font=\fontsize{7.5}{6}\selectfont,
draw = none, 
at={(0.999,0.999)},anchor=north east,nodes={scale=0.8}}, legend cell align={left},
legend entries={\textbf{\# of Edges},\textbf{\# of Nuclei}},
]
\addplot[fill = dbscan] coordinates {  
(1,504) (5,407) (10,229) (15,51)
};
\addplot[fill = blue] coordinates { 
(1,210) (5,25) (10,9) (15,6) 
};

\end{axis}
\end{tikzpicture}

}
\vspace{-0.2cm}
\caption{\textbf{Average PD and PPC, average number of edges, average number of $\l$-$\pmb{(k,\theta)}$-nuclei for \textit{flickr} with $\pmb{\theta=0.3}$.} }
\label{avedensflickr}
\vspace{-0.2cm}
\end{figure}
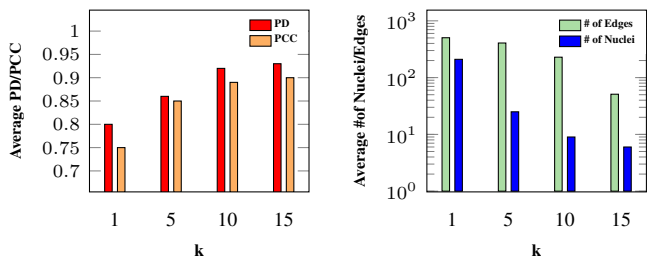

\balance
\subsection{Quality Evaluation of Nucleus Subgraphs}\label{density}

\looseness=-1
Here we compare the cohesiveness of $\l$-$(k,\theta)$-nucleus with the cohesiveness of local $(k,\gamma)$-truss~\cite{huang2016truss} and $(k,\eta)$-core~\cite{bonchi2014core}. We use two metrics. The first metric is the probabilistic density (PD) of a graph $\G$, which we denote by $\text{PD}(\G)$ and is defined as follows~\cite{huang2016truss}:
\begin{equation}
    \text{PD}(\G) = \frac{\sum_{e \in E} p(e)}{\frac{1}{2} \left | V \right | \cdot (\left | V \right |-1)}.
\end{equation}
\looseness=-1
In words, 
PD of a probabilistic graph is the ratio of the sum of edge probabilities to the possible number of edges in a graph. 

The second metric is probabilistic clustering coefficient (PCC). It measures the level of tendency of the nodes to cluster together. 
Given a probabilistic graph $\G$, its PCC is defined as follows~\cite{huang2016truss,pfeiffer2011methods}: 
\begin{equation}
    \text{PCC}(\G) =  \frac{3 \sum_{\t_{uvw} \in \G}p(u,v) \cdot p(v,w) \cdot p(u,w)}{\sum_{(u,v),(u,w), v \neq w }p(u,v) \cdot p(u,w)}.
\end{equation}

\looseness=-1
For probabilistic nucleus, probabilistic truss and probabilistic core subgraphs, we use the same threshold $\theta=\gamma=\eta$, set to 0.1 and 0.3. ($\gamma$ is used as threshold in the truss case \cite{huang2016truss}, and $\eta$ is used as threshold in the core case~\cite{bonchi2014core}).
Table~\ref{densitycompare} reports results on \textit{dblp}, \textit{pokec}, and \textit{biomine}. 
Results for the other datasets 
are similar. 
For a given threshold, we report the statistics of local $(k_{Nmax},\theta)$-nucleus, $(k_{Tmax},\gamma)$-truss, and $(k_{Cmax},\eta)$-core, 
where $k_{Nmax}$, $k_{Tmax}$, and $k_{Cmax}$ are maximum nucleus, truss and core scores, respectively. Also, for $k_{Nmax}$, $k_{Tmax}$, and $k_{Cmax}$, we might obtain more than one connected component; we report the average statistics over such components. 
We denote 
by $V_N,V_T,V_C$, the sets of nodes,
by $E_N,E_T,E_C$, the sets of edges, 
by $PD_N,PD_T,PD_C$, the PD's
and by $PCC_N,PCC_T,PCC_C$, the PCC's
of nucleus, truss, and core components, respectively. \color{black}{The last column shows the running time for computing each decomposition.
We observe that sometimes nucleus decomposition is faster than truss decomposition. This is because in nucleus decomposition there could be fewer triangles that survive the specified threshold in terms of support than edges in truss decomposition. 
}

As can be seen in the table, $(k_{Nmax},\theta)$-nucleus produces high quality results in terms of PD and PCC. 
We achieve a significantly higher PD and PCC for nucleus compared to truss and core. 
For instance, for {\em dblp} the PD for nucleus is 0.8 versus 0.611 and 0.264 for truss and core, which translates for nucleus being about 30\% and 200\% more dense than truss and core. 
Similar conclusions can be drawn for PCC as well. 


\looseness=-1
Moreover, Figure~\ref{avedensflickr} reports the average PD, average PCC, average edges in each $\l$-$(k,\theta)$-nucleus, and number of connected components ($\l$-$(k,\theta)$-nuclei) for an example dataset \textit{flickr} with fixed $\theta=0.3$ and varying $k$. 
We see that even for small values of $k$, PD and PCC are considerably high (above 70-80\%). 
In general, PD and PCC become larger as $k$ increases, since denser nuclei will be detected by removing triangles having low support probability to be part of a $4$-clique. 
This causes the final subgraphs to have edges with high probability only. Furthermore, since $\l$-$(k,\theta)$-nucleus implies connectivity, the number of connected components increases as $k$ decreases. It results in an increase in the average number of edges in each $\l$-$(k,\theta)$-nucleus. 

Finally, we compare the PD and PCC values of 
$\g$-$(k,\theta)$-nucleus, $\w$-$(k,\theta)$-nucleus \color{black}{over 5 runs of these algorithms}, and 
$\l$-$(k,\theta)$-nucleus, for 
\textit{krogan}, \textit{flickr}, and \textit{dblp} datasets using $\theta = 0.001$, and averaging over all the possible values of $k$. 
The results are shown in Figure~\ref{densityglob}. 
We see that $\g$-$(k,\theta)$-nucleus achieves higher cohesiveness as expected. 
In addition, $\w$-$(k,\theta)$-nucleus exhibits quite good PD and PCC values higher than those for $\l$-$(k,\theta)$-nucleus.

\begin{table}
  \centering
    \begin{tabular}{rrrrrrrrrrr}
    $n$ & \multicolumn{2}{c}{AV(PD)} & \multicolumn{2}{c}{AV(PCC)} & \multicolumn{2}{c}{AV(Edge)} & \multicolumn{2}{c}{AV(Vertex)} & \multicolumn{1}{l}{$\epsilon$} & \multicolumn{1}{l}{$\delta$} \\
    \midrule
    \multicolumn{1}{r}{150} & .905 & .726 & .903 & .770 & 12.744 & 55.631 & 5.336 & 11.971 & 0.1   & 0.1 \\
    \multicolumn{1}{r}{300} & .906 & .733 & .903 & .773 & 12.725 & 52.960 & 5.334 & 11.543 & 0.07  & 0.05 \\
    \multicolumn{1}{r}{500} & .906 & .729 & .903 & .767 & 13.005 & 53.883 & 5.383 & 11.703 & 0.05  & 0.06 \\
    \multicolumn{1}{r}{1000} & .905 & .725 & .902 & .766 & 12.823 & 53.772 & 5.356 & 11.745 & 0.05  & 0.01 \\
    \multicolumn{1}{r}{2000} & .906 & .727 & .903 & .768 & 12.782 & 54.264 & 5.350 & 11.792 & 0.03  & 0.05 \\
    \midrule
    AV & .906 & .728 & .903 & .769 & 12.816 & 54.102 & 5.352 & 11.751 &       &  \\
    SD    & .0004 & .003 & .0003 & .002 &    .112   & .978 & .020 & .155 &       &  \\
    \bottomrule
    \end{tabular}%
 \caption{\small Effect of sample size ($n$), $\epsilon$, and $\delta$ on different average metrics,  average PD, average PCC, average number of edges, and average number of vertices for global and weakly-global nuclei. The first and second columns for each metric are for global and weakly-global nuclei, respectively. The results shown here are on krogan with $\theta=0.1$. Observe that standard deviation (SD) is not more than 1.8\% of the average for all columns. For some of the columns SD is much smaller, e.g. for average PD (first column) it is only 0.05\%.}
 \vspace{-0.6cm}
  \label{epsilondelta}%
\end{table}%

\smallskip
\noindent
\textbf{Effect of $\epsilon$ and $\delta$.} We consider \textit{krogan} dataset with $\theta=0.1$. The choice of $\epsilon$ and $\delta$ influence the number $n$ of possible worlds we sample. For $\epsilon=0.1$ and $\delta=0.1$ we obtain $n=150$. In order to see the fidelity of our results, we experiment by increasing $n$ to higher values, namely $300, 500, 1000, 2000$. As the results in Table~\ref{epsilondelta} show, the following metrics about global and weakly-global nuclei: average PD, average PCC, average number of vertices, and average number of edges change very little.
Specifically, the first two metrics are dispersed by not more than $0.4\%$ around their mean over the different values of $n$, and the last two metrics are dispersed by not more than $1.8\%$.
There can be many $\epsilon$ and $\delta$ values corresponding to a given sample size; for illustration, for $n=150$, we can have $\epsilon=0.1$, $\delta=0.1$, whereas for $n=2000$, we can have $\epsilon=0.03$, $\delta=0.05$,
i.e. we see that even though in the latter case the $\epsilon$ and $\delta$ decrease by a factor of 3 and 2, respectively, still the nuclei results in terms of the aforementioned metrics are almost the same.
This validates the choice of $\epsilon$ and $\delta$ to 0.1 since lower values do not offer significant improvement in the quality of results.

\begin{figure}
    \centering
    \subfloat{
    \begin{tikzpicture}
\begin{axis}[
ybar=0.1cm, 
width=4.3cm,
height=4.3cm,
bar width = 0.1cm,
legend style={font=\fontsize{6}{6}\selectfont,
draw = none, 
at={(0.01,0.999)},anchor=north west,nodes={scale=0.8}}, legend cell align={left},
legend entries={{$\g$-$(k,\theta)$-nucleus}, {$\w$-$(k,\theta)$-nucleus},{$\l$-$(k,\theta)$-nucleus}},
enlargelimits=0.3,
enlarge y limits  = 0.3,
ytick align=inside,
ytick distance =0.2,
xtick style={draw=none},
ymin=0.4,
ymax = 1,
ylabel={\textbf{PD}},
ylabel near ticks,
ylabel style={font=\fontsize{6.5}{6}\selectfont},
symbolic x coords={krogan, flickr,dblp},
xtick=data,
xlabel style = {font=\fontsize{6.5}{6}\selectfont},
x tick label style={/pgf/number format/1000 sep=},
ticklabel style = {font=\fontsize{6}{6}\selectfont},
every node near coord/.append style={align=left, rotate=60, anchor=west,font=\fontsize{6}{6}\selectfont},
point meta=rawy,
]

\addplot[fill = red] coordinates {
(krogan,0.72) (dblp, 0.61) (flickr, 0.57) 
%
};
\addplot[fill = blue ] coordinates {
(krogan,0.66) (dblp, 0.5) (flickr, 0.49) 
%
};
\addplot[fill = green] coordinates {
(krogan, 0.6) (flickr, 0.4) (dblp,0.299) 
};

\end{axis}
\end{tikzpicture}
}
\subfloat{
\begin{tikzpicture}
\begin{axis}[
ybar=0.1cm, 
width=4.3cm,
height=4.3cm,
bar width = 0.1cm,
legend style={font=\fontsize{6}{6}\selectfont,
draw = none, 
at={(0.01,0.999)},anchor=north west,nodes={scale=0.8}}, legend cell align={left},
legend entries={{$\g$-$(k,\theta)$-nucleus}, {$\w$-$(k,\theta)$-nucleus},{$\l$-$(k,\theta)$-nucleus}},
enlargelimits=0.3,
enlarge y limits  = 0.3,
ymin = 0.4,
ymax=1,
ytick align=inside,
ytick distance =0.2,
xtick style={draw=none},
ylabel={\textbf{PCC}},
ylabel near ticks,
ylabel style={font=\fontsize{6.5}{6}\selectfont},
symbolic x coords={krogan, flickr,dblp},
xtick=data,
xlabel style = {font=\fontsize{6.5}{6}\selectfont},
x tick label style={/pgf/number format/1000 sep=},
ticklabel style = {font=\fontsize{6}{6}\selectfont},
every node near coord/.append style={align=left, rotate=60, anchor=west,font=\fontsize{6}{6}\selectfont},
point meta=rawy,
]

\addplot[fill = red] coordinates {
(krogan,0.75) (flickr, 0.57) (dblp, 0.612) 
%
};
\addplot[fill = blue] coordinates {
(krogan,0.7) (flickr, 0.48) (dblp, 0.49) 
%
};
\addplot[fill = green] coordinates {
(krogan, 0.6) (flickr, 0.4) (dblp,0.32) 
};

\end{axis}
\end{tikzpicture}
}
\vspace{-0.3cm}
\caption{PD and PCC for $\g$-$\pmb{(k,\theta)}$, $\w$-$\pmb{(k,\theta)}$, and 
$\l$-$\pmb{(k,\theta)}$ nuclei on \textit{krogan}, \textit{flickr}, and \textit{dblp}.
}
\label{densityglob}
\end{figure}
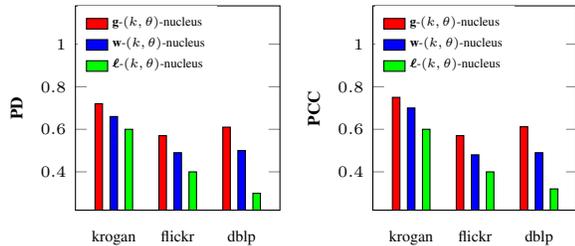


\vspace{-0.2cm}
\looseness=-1
\color{black}{\subsection{Case Study} \label{casestudy}}
\color{black}{\textbf{Analysis of DBLP Collaboration Network for task-driven team formation.} To show the usefulness of nucleus decomposition in probabilistic graphs, we apply our decomposition algorithms to solve the \textit{task-driven team formation} problem for a DBLP network. In task-driven team formation \cite{bonchi2014core}, we are given a probabilistic graph $ \G^T = (V,E,p^T)$, which is particularly obtained for task $T$. Vertices in $\G^T$ are individuals and edge probabilities are obtained with respect to task $T$ as described in~\cite{bonchi2014core}. 
Given a query $\langle Q,T \rangle$, where $Q \subset V$, and $T$ is a set of keywords describing a task, the goal is to find a set of vertices that contain $Q$ and make a good team to perform the task described by the keywords in $T$. 
By a good team we mean a good affinity among the team members in terms of collaboration for the given task. 
To solve task-driven team formation using nucleus decomposition, we extend the definition of \cite{bonchi2014core} to employ probabilistic nucleus: 
Given a probabilistic graph $\G^T = (V,E,p^T)$ with respect to a task $T$, a query set $Q$ of vertices, and a threshold $\theta$, apply nucleus decomposition on $\G^T$ and find a $(k,\theta)$-nucleus (local/weakly-global/global) which \textbf{(1)} contains the vertices in $Q$, and \textbf{(2)} has the highest value of $k$ for the given $\theta$, and return the obtained subgraph as a solution.}

\looseness=-1
\color{black}{In our experiment, we use a DBLP collaboration network from~\cite{bonchi2014core}, where vertices are authors, and edges represent collaboration on at least one paper. 
The dataset has $1,089,442$ vertices and $4,144,697$ edges. For each edge, we take the bag of words of the title of all papers coauthored by the two authors connected by the edge and apply Latent Dirichlet Allocation (LDA)~\cite{blei2003latent,bonchi2014core} to infer its topics and calculate the edge probability. 
Given a task $T$ with keywords, and the input collaboration network, we obtain a probabilistic graph $\G^T$, in which $p(u,v)$ represents the collaboration level in the papers co-authored by $u$ and $v$ related to task $T$~(\cite{bonchi2014core,huang2016truss}). 
}

The first sample query we consider is 
$\langle
\{\mbox{``algorithm''}\},$ \\
$\{\mbox{``Erik\_D.\_Demaine''}$, $\mbox{``J.\_Ian\_Munro''}, \allowdisplaybreaks \mbox{``John\_Iacono''}\}\rangle$. Figure~\ref{taskdrivenex1a} shows the subgraph obtained by $\l$-$(k,\theta)$-nucleus and $\w$-$(k,\theta)$-nucleus decompositions, where $k=2$ and $\theta = 10^{-11}$. 
{\color{black}
The threshold is the same as the ones used in case studies of previous works (on truss and core).
}
As discussed in~\cite{bonchi2014core}, the edge probabilities in the data are very small, which requires setting threshold $\theta$ to a small value. 

{
\color{black}
\textbf{Remark}. It should be noted that picking an appropriate value for the threshold can be done using binary search over $(0,b]$, where $b\leq 1$.
}
The subgraph contains all the three authors in the query. It has $10$ vertices and $33$ edges. As can be seen, the obtained subgraph is quite good for task-driven team formation. All the authors in the subgraph are well-known and have strong collaboration affinity to work on a research paper related to algorithms. A $\g$-$(k,\theta)$-nucleus (same $k$ and $\theta$) that contains the query vertices is shown with thick blue edges in the same figure. As expected, this subgraph is more cohesive and it happens to be a clique of size $6$. Its density and clustering coefficient (PCC) is $0.138$ and $0.140$ as opposed to $0.099$ and $0.110$ for the local and weakly-global subgraphs.  From a research perspective the collaborations of the academicians in the blue subgraph are more focused on designing efficient data-structures.

{\color{black}
We run the global truss algorithm on the dataset. As expected the global truss subgraph which contains the query authors is bigger than global nucelus (9 vertices and 18 edges) and its density and PCC are lower (0.067 and 0.086).

We also run global core decomposition as in~\cite{peng2018efficient} for the same value $k$ and $\theta$. It should be noted that the global definition is different from global truss and global nucleus. Also, it does not assume connectivity between nodes. However, for fairness of comparison, we considered a connected component of this subgraph which contains query authors. The obtained subgraph contains $569$ vertices and $5294$ edges, with density $0.003$ and PCC $0.061$.  

Regarding local truss, we obtained a subgraph with 170 vertices and 1033 edges with density equal to $0.008$ and PCC equal to $0.0872$. On the other-hand, local core decomposition results in density and PCC being equal to $0.0084$ and $0.0659$ with 226 vertices and 2631 edges. As can be seen, our nucleus decomposition algorithm results in much better subgraphs in terms of vertex size and cohesiveness.

}

\begin{figure}[h]
    \centering
    \subfloat[]{\label{taskdrivenex1a}
    \includegraphics[scale=0.08]{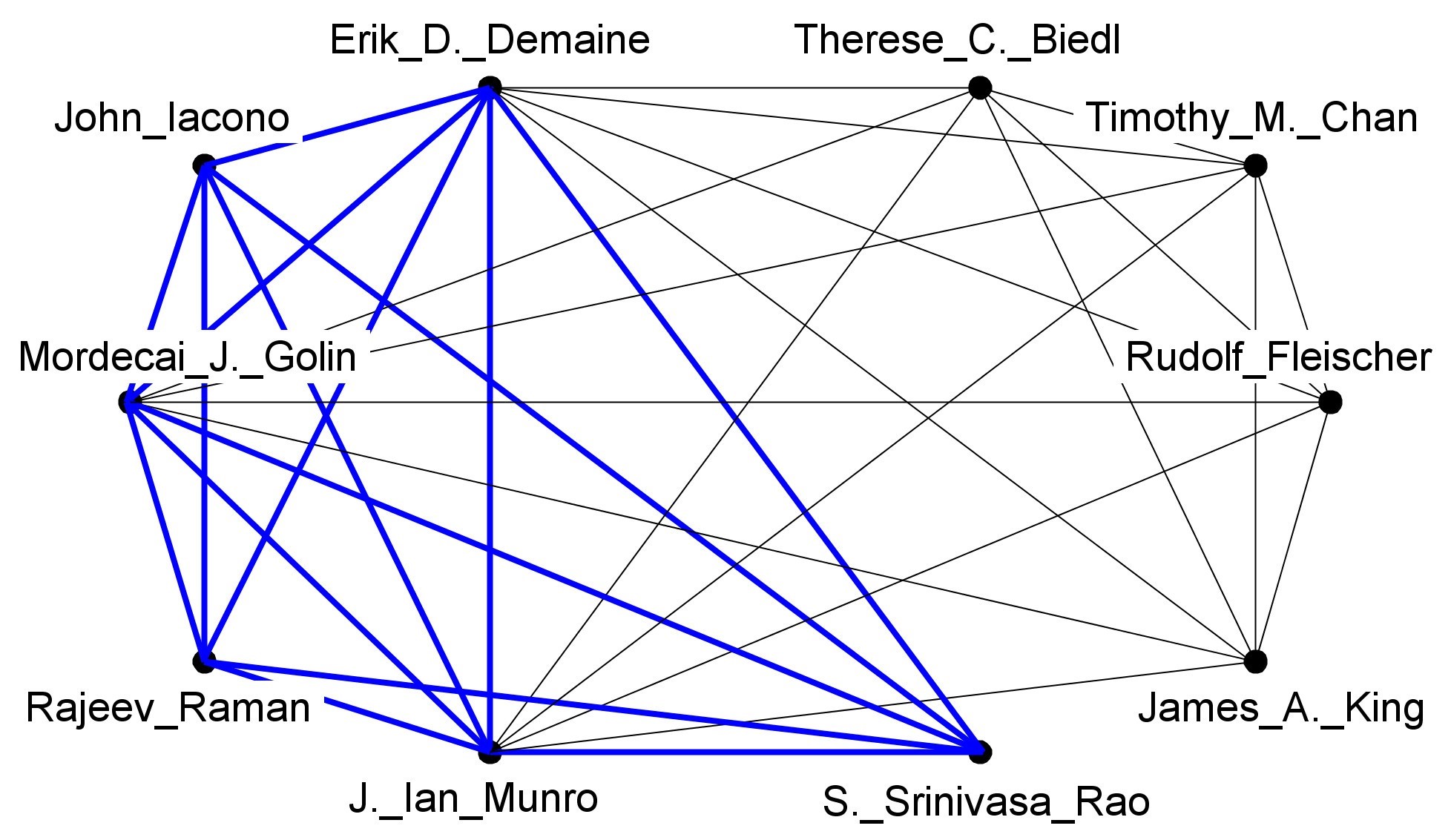}
    }
    \subfloat[]{\label{taskdrivenex2}
    \includegraphics[scale=0.05]{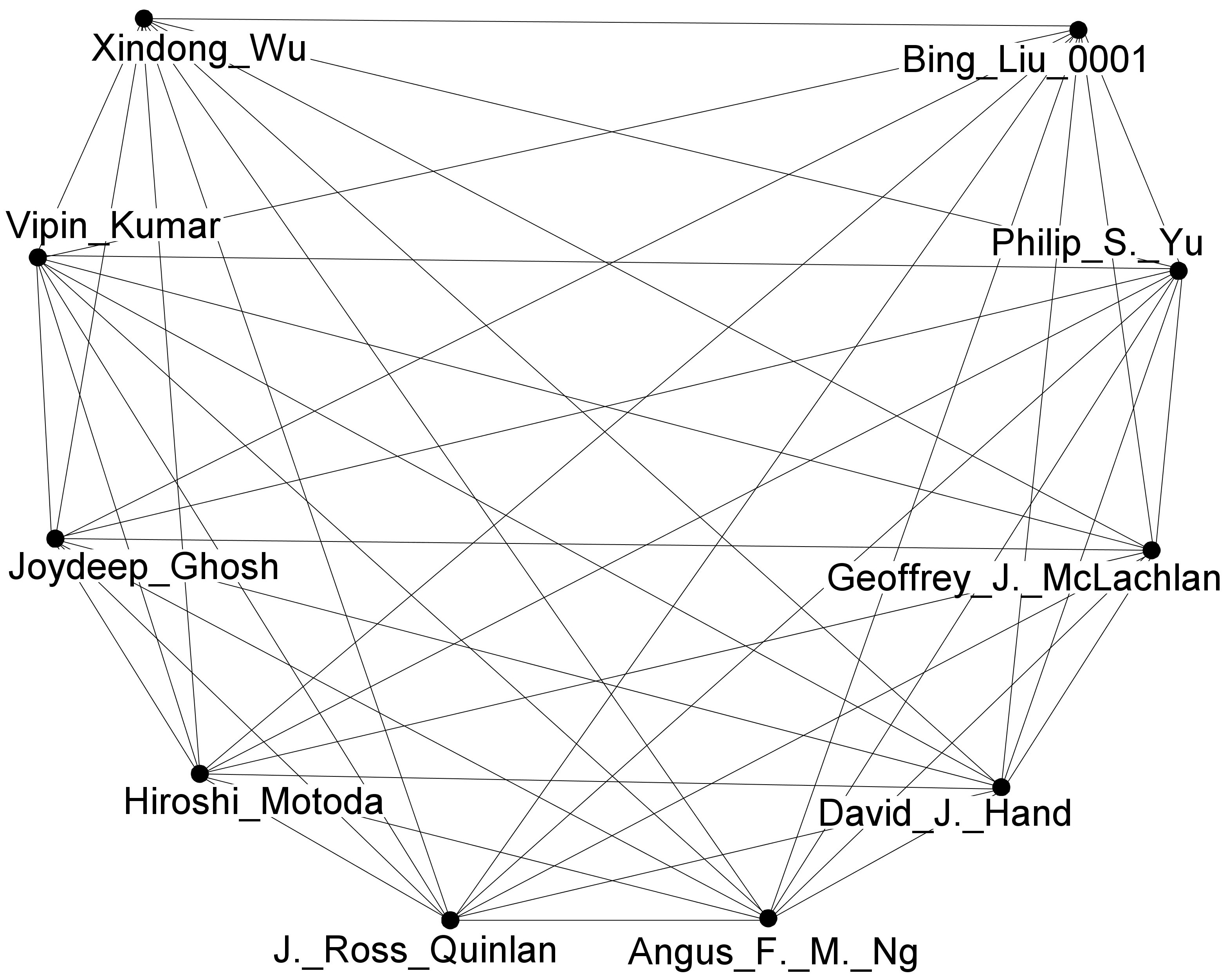}
    }
    \vspace{-0.2cm}
    \caption{\small \textbf{a)} A case study of task-driven team formation with keyword \{``algorithm''\}, and query vertices \{``Erik\_D.\_Demaine'', ``J.\_Ian\_Munro'', ``John\_Iacono''\}, $k=2$, and $\theta = 10^{-11}.$ The depicted graph with thick blue edges corresponds to a $\g$-$(k,\theta)$ nucleus. The whole graph (of 10 vertices) is a $\l$-$(k,\theta)$ nucleus which coincides with a $\w$-$(k,\theta)$ nucleus in this example. 
    \textbf{b)} A weakly-global $\w$-$(k,\theta)$ nucleus for task-driven team formation with query nodes \{``Xindong\_Wu'', ``Bing\_Liu\_0001'', ``Vipin\_Kumar''\}, and keyword \{ ``algorithm''\}. $k=1$, and $\theta = 10^{-11}.$}
    \vspace{-0.1cm}
\end{figure}

\looseness=-1

The second query we use shows the usefulness of the weakly-global notion.  
It has keyword \{``algorithm''\} and vertices \{``Xindong\_Wu'',  ``Bing\_Liu\_0001'', ``Vipin\_Kumar''\}. 
Figure~\ref{taskdrivenex2} shows the $\w$-$(k,\theta)$ nucleus for this query, where the threshold is the same as before, and $k=1$. 
The local nucleus containing the query authors had more than 100 nodes while the global nucleus containing these three query authors was empty. This example shows that the weakly global notion can discover interesting teams when the other two notions produce teams that are too big or too small (or empty). In particular, all the authors in the resulting subgraph are very well-know and have similar research area which can form a good team related to keyword algorithm (query keyword). 
{
\color{black}{
On the other-hand, using global truss decomposition we could not obtain any subgraph. In addition, both local truss and core decompositins, did not lead to a desired team as the number of vertices in such graphs is very large, $16663$ and $31300$, respectively. In fact, it is not realistic for this amount of authors to collaborate on paper related to algorithm. The density and PCC for weakly-global subgraph is $0.036$ and $0.0388$, as opposed to density $0.00005$ and PCC $0.0280$ in local truss and density $0.000001$ and PCC $0.0236$ in local core. The same argument hold for global core with $2997$ vertices, $35354$ edges, density $0.0004$, and PCC $0.0294$. For local nucleus decomposition cohesiveness results show density $0.03$ and and PCC $0.0331$ with vertices $100$ which is much smaller than local core and local truss.
}
}

\color{black}{Compared to other notions of dense subgraphs in probabilistic graphs, such as truss decomposition of \cite{huang2016truss}, we observed that our nuclei notions capture denser subgraphs better than the truss counterparts. 
For instance, in the example of \cite{huang2016truss} for task-driven query of 
$\langle  
\{\mbox{``data'', ``algorithm''}\},
\{\mbox{``Jeffrey\_D.\_Ullman''},  \mbox{``Piotr\_Indyk''}\}
\rangle$, local nucleus gives a smaller community than local truss, namely the community obtained by the global truss: 
(``Jeffrey\_D.\_Ullman'', ``Shinji\_Fujiwara'', ``Aristides\_Gionis", ``Rajeev\_Motwani'', ``Mayur\_Datar", ``Edith\_Cohen'', ``Cheng\_Yang'', ``Piotr\_Indyk'').
This is interesting as it shows that, in some cases, communities obtained by the exponential time algorithm of \cite{huang2016truss} for global truss can be obtained by our polynomial time algorithm for local nucleus.}

\color{black}{
In summary, our case study shows that we can discover good communities with reasonable cohesiveness using the efficient algorithm for local nucleus. However, some local nuclei can be too big. If so, we can apply the algorithms for weakly-global or global nucleus decomposition on the local nuclei to get smaller and denser communities. 
}

{

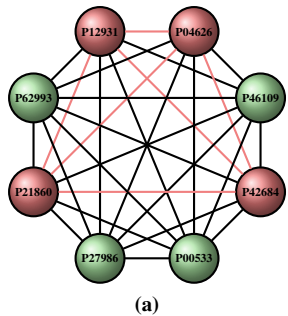
\begin{figure}
    \centering
    \subfloat[]{ \label{ppi_1}
     \begin{tikzpicture}[auto, node distance=2.5cm, every loop/.style={},
                    thick,main node/.style={scale=0.5, 
                    circle, ball color = grannysmithapple, draw,font=\fontsize{9}{6}\selectfont\bfseries}]
                    \node[main node,shading = ball, ball color = lightcoral] (0) {P12931};
                    \node[main node,shading = ball, ball color = lightcoral] (1) [right of = 0] {P04626};
                    \node[main node] (7) [below left of = 0] {P62993};
                    \node[main node] (2) [below right of = 1] {P46109};
                    \node[main node,shading = ball, ball color = lightcoral] (3) [below of = 2] {P42684};
                    \node[main node,shading = ball, ball color = lightcoral] (6) [below of = 7] {P21860};
                    \node[main node] (5) [below right of = 6] {P27986};
                    \node[main node] (4) [below left of = 3] {P00533};
                    
                      \path[every node/.style={
                      font=\sffamily\small}]
                      (0) edge[lightcoral] node [above, font=\fontsize{8}{6}\selectfont\bfseries] {} (1)
                      edge[] node [left, font=\fontsize{8}{6}\selectfont\bfseries] {} (2) 
                       edge[lightcoral] node [left, font=\fontsize{8}{6}\selectfont\bfseries] {} (3)
                        edge[] node [left, font=\fontsize{8}{6}\selectfont\bfseries] {} (4)
                         edge[] node [left, font=\fontsize{8}{6}\selectfont\bfseries] {} (5)
                          edge[lightcoral] node [left, font=\fontsize{8}{6}\selectfont\bfseries] {} (6)
                           edge[] node [left, font=\fontsize{8}{6}\selectfont\bfseries] {} (7)
                            
                        (1) edge[] node [above, font=\fontsize{8}{6}\selectfont\bfseries] {} (2)
                        (1) edge[lightcoral] node [above, font=\fontsize{8}{6}\selectfont\bfseries] {} (3)
                        (1) edge[] node [above, font=\fontsize{8}{6}\selectfont\bfseries] {} (4)
                        (1) edge[] node [above, font=\fontsize{8}{6}\selectfont\bfseries] {} (5)
                        (1) edge[lightcoral] node [above, font=\fontsize{8}{6}\selectfont\bfseries] {} (6)
                        (1) edge[] node [above, font=\fontsize{8}{6}\selectfont\bfseries] {} (7)
                        
                        (2) edge[] node [above, font=\fontsize{8}{6}\selectfont\bfseries] {} (3)
                        (2) edge[] node [above, font=\fontsize{8}{6}\selectfont\bfseries] {} (4)
                        (2) edge[] node [above, font=\fontsize{8}{6}\selectfont\bfseries] {} (5)
                        (2) edge[] node [above, font=\fontsize{8}{6}\selectfont\bfseries] {} (6)
                        (2) edge[] node [above, font=\fontsize{8}{6}\selectfont\bfseries] {} (7)
                        
                        (3) edge[] node [above, font=\fontsize{8}{6}\selectfont\bfseries] {} (4)
                        (3) edge[] node [above, font=\fontsize{8}{6}\selectfont\bfseries] {} (5)
                        (3) edge[lightcoral] node [above, font=\fontsize{8}{6}\selectfont\bfseries] {} (6)
                        (3) edge[] node [above, font=\fontsize{8}{6}\selectfont\bfseries] {} (7)
                        
                        (4) edge[] node [above, font=\fontsize{8}{6}\selectfont\bfseries] {} (5)
                        (4) edge[] node [above, font=\fontsize{8}{6}\selectfont\bfseries] {} (6)
                        (4) edge[] node [above, font=\fontsize{8}{6}\selectfont\bfseries] {} (7)
                        
                        (5) edge[] node [above, font=\fontsize{8}{6}\selectfont\bfseries] {} (6)
                        (5) edge[] node [above, font=\fontsize{8}{6}\selectfont\bfseries] {} (7)
                        
                        (6) edge[] node [above, font=\fontsize{8}{6}\selectfont\bfseries] {} (7);
                    
    \end{tikzpicture}
     }


    \caption{\color{black}$\w$-$(k,\theta)$ nucleus (green and pink nodes) and a $\g$-$(k,\theta)$-nucleus (pink nodes) which contain protein nodes of interest; \textsf{P04626}, \textsf{P12931}, \textsf{P42684}, where $k=1$ and $\theta=0.001$.}
     \label{ppiexperiment}
 \end{figure}   
}
{\color{black}
\textbf{Nucleus Decomposition on the Human Biomine Dataset.}

We use the human biomine dataset~\cite{biomine2019}, which has 861,812 nodes and 8,666,287 edges. This dataset is different  from the biomine dataset we used for our efficiency evaluation. 
We consider how our notions perform in detecting proteins/genes that interact with the \textbf{SARS-CoV-2} coronavirus. 
Bouhaddou et al.~\cite{bouhaddou2020global} found that during the \textbf{SARS-CoV-2} virus infection, changes in activities can happen for human kinases. 
We select three proteins, \textsf{P04626}, \textsf{P12931} and \textsf{P42684}; they are tyrosine kinase-related proteins and come from UniProt, which is a freely accessible database of protein sequences and functional information. 
The gene names associated with these proteins are \textsf{SRC}, \textsf{ERBB2}, and \textsf{ABL2}. 
These proteins have received literature support for interaction with \textbf{SARS-CoV-2} coronavirus~
\cite{marchetti2020covid,zheng2020examining,taniguchi2021increased,guo2021integrative,zhao2020imatinib,ebrahimi2021interferon,bouhaddou2020global}.
We refer to them as proteins of interest. 
We find the subgraphs obtained by local, weakly-global, and global nucleus decomposition which contain these three nodes. 
Moreover, at the same time we compare these graphs with their counterparts, truss and core in terms of density and size of the subgraph. For all the notions we set threshold $\theta = 0.001$.

\begin{table}
\color{black}
  \centering
    \begin{tabular}{lrrrr}
        Notion  &  
        \multicolumn{1}{l}{Max $k$} & \multicolumn{1}{l}{Nodes} & \multicolumn{1}{l}{Density} \\
        \hline
    l-core     & 88    & 2408  & 0.04 \\
    g-core     & 31    & 10026   & 0.01 \\
    l-truss     & 4     & 5787  & 0.01 \\
    l-nucleus   & 1     & 95    & 0.06 \\
    \textbf{g-truss}    & 2     & 10    & \textbf{0.44} \\
    \textbf{w-nucleus}    & 1     & 8     & \textbf{0.51} \\
    \textbf{g-nucleus} & 1     & 4     & \textbf{0.56} \\
    \hline
    \end{tabular}%
  \caption{\color{black}Comparison of different dense subgraph notions with respect to (1) max $k$ for which the subgraph contains the proteins of interest,  
  (2) number of nodes in the subgraph, and (3) density of the subgraph. Parameter $\theta$ is set to 0.001 for all the notions.
  We see that l-nucleus is denser than l-truss and both l-core and g-core. Also, w-nucleus and g-nucleus are denser than g-truss. In terms of nodes, l-nucleus gives a subgraph which is much smaller than the subgraphs of l-core, g-core, and l-truss. Such a graph of 95 nodes is more amenable for further processing by human analysts.}
   \label{cohesiveness}%
\end{table}%

Table~\ref{cohesiveness} shows the comparison of different dense subgraph notions with respect to (1) largest $k$ for which the subgraph contains the proteins of interest,  
(2) number of nodes in the subgraph, and (3) density of the subgraph.
We see that l-nucleus is denser than l-truss and both l-core and g-core. Also, w-nucleus and g-nucleus are denser than g-truss. In terms of nodes, l-nucleus gives a subgraph which is much smaller than the subgraphs of l-core, g-core, and l-truss. 
More precisely, with respect to l-nucleus, the three proteins of interest appear in a nucleus of 95 vertices and 509 edges. 
To see which kind of biology function/process our detected community represent, we use \textit{Metascape} (\url{ https://metascape.org/gp/index.html#/main/step1}). Metascape~\cite{zhou2019metascape} is a web-based portal that provides comprehensive gene list annotation and analysis resources. 
Using Metascape, we find that the proteins in the local nucleus are associated with several diseases, most of them being forms of cancer (16 out 20). The p-values of the association are less than $10^{-18}$, which is statistically very significant.  

Figure~\ref{ppiexperiment} shows the weakly global and global nuclei which contain the proteins of interest. All the nodes (green and pink) comprise the weakly global subgraph. The pink nodes comprise the global nucleus subgraph.
Using Metascape, we find that the proteins in our weakly-global and global subgraphs are associated with some more specific forms of cancer such as \textit{Uterine Carcinosarcoma} and \textit{Hormone Refractory Prostate Cancer}, respectively, with p-values less than $10^{-6}$, which are statistically quite significant, especially given the fact that these subgraphs are much smaller than the local nucleus (in general, the more observations we have, the smaller the p-values become).   
These findings are useful to biologists in order to perform targeted tests for checking whether drugs for the treatment of these diseases can also be repurposed for treating COVID-19~\cite{guo2021integrative}. There are over 250 anticancer drugs approved by the FDA, but far fewer for specific kinds of cancer. Thus, showing connections to specific forms helps narrow the choice of drugs to repurpose.

In summary, it is running all the three versions of nucleus decomposition on the Biomine dataset that gives surprising subgraphs pointing to potentially useful further investigation by biologists. Running only local nucleus decomposition will miss such interesting groups, no matter how we set the values of $k$ and $\theta$.
}
\begin{figure*}
    \subfloat[]{
    \includegraphics[scale=1.2,width=2.4in]{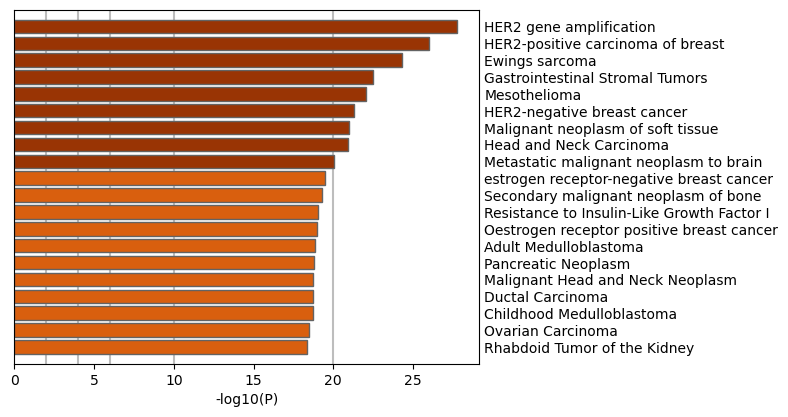}
    } 
    \subfloat[]{
    \includegraphics[scale=1.2,width=2.4in]{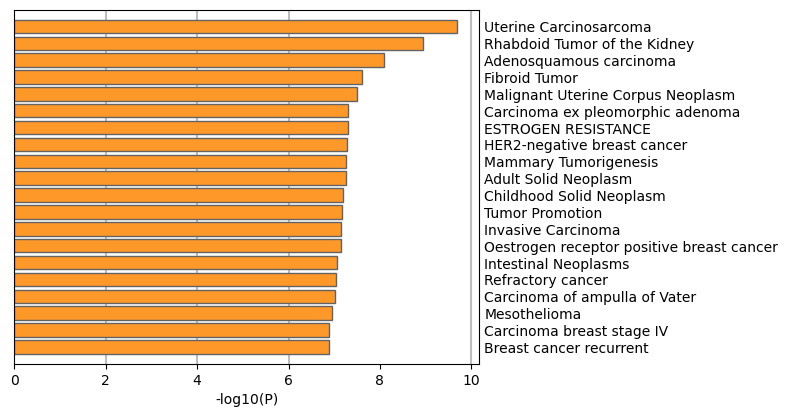}
    }
    \subfloat[]{
    \includegraphics[scale=1.2,width=2.4in]{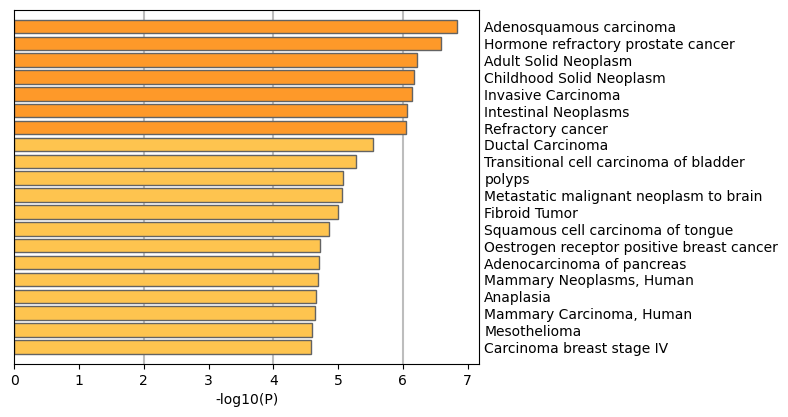}
    }
    \caption{\color{black} Top enriched terms related to diseases in the detected subgraphs by local, weakly-global, and global nucleus decompsoitions. Variable \textbf{P} on the $x$-axis refers to p-value. }
    \label{ppiexp3}
\end{figure*}

{
\color{black}
\textbf{Discussion of the difference between weakly-global and global definitions on \textit{BrightKite}. }
To show more applications on the difference between weakly-global and global definitions of nucleus decomposition in probabilistic graphs, we consider social network data from \textit{BrightKite} (\url{https://snap.stanford.edu/data/loc-brightkite.html}). \\ 
\textit{BrightKite} was once a location-based social networking service provider where users shared their locations by checking-in. The friendship network was collected using their public API, and consists of 58,228 nodes and 214,078 edges, and 4,491,143 checkins between April 2008 and October 2010. 
We generated probabilities for each edge based on the Jaccard similarity between the neighborhoods of two endpoints. 
Running weakly-global and global nucleus decompositions on this dataset with $\theta=0.1$, we retrieve $300$ and $20$ $\g$-$(k,\theta)$ and $\w$-$(k,\theta)$ nuclei, respectively. 
For weakly-global subgraphs, $k$ ranges in $[1,5]$, 
and for global subgraphs, $k$ can take on values of $1$ and $2$.

As expected, global nuclei obtain better cohesiveness in terms of density and clustering coefficient. In particular, the average density and clustering coefficient in global nuclei over all values of $k$, is $0.6951$ and $0.6947$ as opposed to $0.4844$ and $0.5052$ in weakly-global nuclei. 
We also report another interesting observation on this dataset. 
We obtain the average number of checkins by users
in the detected subgraphs. 
The average number of user checkins in global nuclei is about $6 \%$ more than those in weakly-global nuclei. 
Moreover, there exist periods, for instance, the period between August 2008 and April 2009, in which the average number of checkins of the users in the global nuclei is $57\%$ more than the average number of checkins in the weakly-global subgraphs.
These results show that global nuclei can capture better user engagement (as measured by the number of checkins) than weakly-global nuclei.

\noindent
\textbf{Remark.}
we explain that all our three models are useful and they should be used in tandem. 
    Local nucleus helps to identify dense subgraphs of interest. 
    We can adjust $k$ and $\theta$ to obtain smaller and denser subgraphs. However, global and weakly global nuclei can identify pockets that are impossible to obtain with local nucleus no matter how we adjust $k$ and $\theta$. 
    For instance, in Example~\ref{fig:my_label} in the paper, it is only the global nucleus that can identify $H_1$ and $H_2$; no other notion can. 
    In our DBLP use case, the local and weakly-global notions helped us identify a dense subgraph of researchers working on Algorithms, however, the global nucleus gave a particular pocket of researchers, who, after close examination, turned out to be especially focused on designing efficient data-structures. 
    Then in the same case study, we were able to identify a useful weakly-global nucleus of five researchers, who are well known to work on algorithms for data mining. The local nucleus was too big (more than 100 nodes), whereas the global nucleus was empty. 
    All these examples show that an analyst should run all the three versions of nucleus decomposition in tandem on a dataset and then closely examine the results. 
    We stress out that this is not just to obtain denser subgraphs as we go from local to weakly-global and global. More than density, what is important is the detection of small pockets of nodes with nice properties that escape getting identified by other notions. 
    Finally, in the Biomine dataset, we observe that the group of proteins in a local nucleus containing three proteins of interest were related to many forms of cancer even though the proteins of interest have received literature support related to COVID-19. Based on consultations with Bioinformatics researchers, this finding is of great importance in finding relationships between seemingly distant diseases. 
    Regarding weakly-global and global notions, they were able to find subgraphs of the local nucleus that were comprised of proteins related to more specific cancer diseases. Investigating the connection of these diseases to COVID-19 is an interesting avenue to explore for a biologist researcher. 
    To reiterate, a researcher should use all three notions of nucleus decomposition as they provide different different view-points and can reveal subgraphs which can be missed by other notions.  

}

\section{Related Work}

In deterministic graphs, core and truss decompositions have been studied extensviely~\cite{seidman1983network,huang2014querying,cohen2008trusses,zhang2012extracting,chen2014distributed,chen2016efficient,montresor2012distributed,zhang2016engagement,cheng2011efficient,sariyuce2013streaming,wang2020efficient,preti2021strud}.
Core decomposition in probabilistic graphs has been studied in~\cite{bonchi2014core,esfahani2019efficient,peng2018efficient,yang2019index}. Bonchi et al.~\cite{bonchi2014core} were the first to introduce core decomposition for such graphs. They focus on finding a subgraph in which each vertex is connected to $k$ neighbors within that subgraph with high probability. 
In~\cite{esfahani2019efficient} 
more efficient algorithms were proposed which can also handle graphs that do not fit in main memory. 
In~\cite{peng2018efficient}, the authors focus on finding a subgraph which contains nodes with high probability to be $k$-core member in the probabilistic graph.
In~\cite{yang2019index}, an index-based structure is defined for processing core decomposition in probabilistic graphs.

In the probabilistic context, the notion of local $(k,\eta)$-truss is introduced by Huang, Lu, and Lakshmanan in \cite{huang2016truss}. Their proposed algorithm for computing local $(k,\eta)$-truss is based on iterative peeling of edges with support less than $k$ and updating the support of affected edges.
Also, \cite{huang2016truss} proposed the notion of global $(k,\eta)$-truss based on the probability of each edge belonging to a $k$-truss in a possible world. 
In~\cite{esfahani2019fast} an approximate algorithm for the local truss decomposition is proposed to efficiently compute the tail probability of edge supports in the peeling process of~\cite{huang2016truss}. In~\cite{sun2021efficient} truss decomposition is computed using an index-based approach.

\looseness=-1
Building on the well-studied notions of core and truss decomposition, Sarıyüce~et~al.~\cite{sariyuce2015finding} introduce nucleus decomposition in \textit{deterministic} graphs. They propose an algorithm for computing $(3,4)$-nuclei.
In a more recent work, Sarıyüce~et~al.~\cite{sariyuce2018local} propose efficient distributed algorithms for nucleus decomposition. 
Our work is the first to study nucleus decomposition in probabilistic graphs. 


\section{Conclusions}
In this work, we made several key contributions. We  introduced the notion of local, weakly-global and global nuclei for probabilistic graphs. We showed that computing weakly-global and global nuclei is intractable. We complemented these hardness results with effective algorithms to approximate them using techniques from Monte-Carlo sampling. 

We designed a polynomial time, peeling based algorithm for computing local nuclei based on dynamic programming and showed that its efficiency can be much improved using novel approximations based on Poisson, Binomial and Normal distributions. Finally, using an in-depth experimental study, we demonstrated the efficiency, scalability and accuracy of our algorithms for nucleus decomposition on real world datasets.

\balance

\bibliographystyle{IEEEtran}
\bibliography{ICDE_FullVersion}

\begin{thebibliography}{10}
\providecommand{\url}[1]{#1}
\csname url@samestyle\endcsname
\providecommand{\newblock}{\relax}
\providecommand{\bibinfo}[2]{#2}
\providecommand{\BIBentrySTDinterwordspacing}{\spaceskip=0pt\relax}
\providecommand{\BIBentryALTinterwordstretchfactor}{4}
\providecommand{\BIBentryALTinterwordspacing}{\spaceskip=\fontdimen2\font plus
\BIBentryALTinterwordstretchfactor\fontdimen3\font minus
  \fontdimen4\font\relax}
\providecommand{\BIBforeignlanguage}[2]{{%
\expandafter\ifx\csname l@#1\endcsname\relax
\typeout{** WARNING: IEEEtran.bst: No hyphenation pattern has been}%
\typeout{** loaded for the language `#1'. Using the pattern for}%
\typeout{** the default language instead.}%
\else
\language=\csname l@#1\endcsname
\fi
#2}}
\providecommand{\BIBdecl}{\relax}
\BIBdecl

\bibitem{bonchi2014core}
F.~Bonchi, F.~Gullo, A.~Kaltenbrunner, and Y.~Volkovich, ``Core decomposition
  of uncertain graphs,'' in \emph{Proceedings of the 20th ACM SIGKDD
  International Conference on Knowledge Discovery and Data Mining}.\hskip 1em
  plus 0.5em minus 0.4em\relax ACM, 2014, pp. 1316--1325.

\bibitem{mukherjee2015mining}
A.~P. Mukherjee, P.~Xu, and S.~Tirthapura, ``Mining maximal cliques from an
  uncertain graph,'' in \emph{2015 IEEE 31st Int. Conf. on Data
  Engineering}.\hskip 1em plus 0.5em minus 0.4em\relax IEEE, 2015, pp.
  243--254.

\bibitem{jin2011discovering}
R.~Jin, L.~Liu, and C.~Aggarwal, ``Discovering highly reliable subgraphs in
  uncertain graphs,'' in \emph{Proceedings of the 17th ACM SIGKDD International
  Conference on Knowledge Discovery and Data Mining}.\hskip 1em plus 0.5em
  minus 0.4em\relax ACM, 2011, pp. 992--1000.

\bibitem{kempe2003maximizing}
D.~Kempe, J.~Kleinberg, and {\'E}.~Tardos, ``Maximizing the spread of influence
  through a social network,'' in \emph{Proceedings of the ninth ACM SIGKDD
  International Conference on Knowledge Discovery and Data Mining}, 2003, pp.
  137--146.

\bibitem{budak2011limiting}
C.~Budak, D.~Agrawal, and A.~El~Abbadi, ``Limiting the spread of misinformation
  in social networks,'' in \emph{Proceedings of the 20th International
  Conference on World Wide Web}, 2011, pp. 665--674.

\bibitem{Tang2014}
Y.~Tang, X.~Xiao, and Y.~Shi, ``Influence maximization: Near-optimal time
  complexity meets practical efficiency,'' in \emph{Proceedings of the 2014 ACM
  SIGMOD International Conference on Management of Data}, 2014, pp. 75--86.

\bibitem{jin2011distance}
R.~Jin, L.~Liu, B.~Ding, and H.~Wang, ``Distance-constraint reachability
  computation in uncertain graphs,'' \emph{Proceedings of the VLDB Endowment},
  vol.~4, no.~9, pp. 551--562, 2011.

\bibitem{zou2010discovering}
Z.~Zou, H.~Gao, and J.~Li, ``Discovering frequent subgraphs over uncertain
  graph databases under probabilistic semantics,'' in \emph{Proceedings of the
  16th ACM SIGKDD international conference on Knowledge discovery and data
  mining}, 2010, pp. 633--642.

\bibitem{goyal2010learning}
A.~Goyal, F.~Bonchi, and L.~V. Lakshmanan, ``Learning influence probabilities
  in social networks,'' in \emph{Proceedings of the third ACM International
  Conference on Web search and Data Mining}, 2010, pp. 241--250.

\bibitem{korovaiko2013trust}
N.~Korovaiko and A.~Thomo, ``Trust prediction from user-item ratings,''
  \emph{Social Network Analysis and Mining}, vol.~3, no.~3, pp. 749--759, 2013.

\bibitem{kuter2010using}
U.~Kuter and J.~Golbeck, ``Using probabilistic confidence models for trust
  inference in web-based social networks,'' \emph{ACM Transactions on Internet
  Technology (TOIT)}, vol.~10, no.~2, p.~8, 2010.

\bibitem{Genome}
G.~Cavallaro, ``Genome-wide analysis of eukaryotic twin cx 9 c proteins,''
  \emph{Molecular BioSystems}, vol.~6, no.~12, pp. 2459--2470, 2010.

\bibitem{dittrich2008identifying}
M.~T. Dittrich, G.~W. Klau, A.~Rosenwald, T.~Dandekar, and T.~M{\"u}ller,
  ``Identifying functional modules in protein--protein interaction networks: an
  integrated exact approach,'' \emph{Bioinformatics}, vol.~24, no.~13, pp.
  i223--i231, 2008.

\bibitem{dong2007understanding}
J.~Dong and S.~Horvath, ``Understanding network concepts in modules,''
  \emph{BMC systems biology}, vol.~1, no.~1, p.~24, 2007.

\bibitem{sharan2007network}
R.~Sharan, I.~Ulitsky, and R.~Shamir, ``Network-based prediction of protein
  function,'' \emph{Molecular systems biology}, vol.~3, no.~1, p.~88, 2007.

\bibitem{zhao2012large}
F.~Zhao and A.~Tung, ``Large scale cohesive subgraphs discovery for social
  network visual analysis,'' \emph{Proceedings of the VLDB Endowment}, vol.~6,
  pp. 85--96, 2012.

\bibitem{zhang2005general}
B.~Zhang and S.~Horvath, ``A general framework for weighted gene co-expression
  network analysis,'' \emph{Statistical applications in genetics and molecular
  biology}, vol.~4, no.~1, 2005.

\bibitem{fratkin2006motifcut}
E.~Fratkin, B.~T. Naughton, D.~L. Brutlag, and S.~Batzoglou, ``Motifcut:
  regulatory motifs finding with maximum density subgraphs,''
  \emph{Bioinformatics}, vol.~22, no.~14, pp. e150--e157, 2006.

\bibitem{fang2020survey}
Y.~Fang, X.~Huang, L.~Qin, Y.~Zhang, W.~Zhang, R.~Cheng, and X.~Lin, ``A survey
  of community search over big graphs,'' \emph{The VLDB Journal}, vol.~29,
  no.~1, pp. 353--392, 2020.

\bibitem{li2020finding}
R.~Li, L.~Qin, F.~Ye, G.~Wang, J.~X. Yu, X.~Xiao, N.~Xiao, and Z.~Zheng,
  ``Finding skyline communities in multi-valued networks,'' \emph{The VLDB
  Journal}, pp. 1--26, 2020.

\bibitem{antiqueira2009complex}
L.~Antiqueira, O.~N. Oliveira~Jr, L.~da~Fontoura~Costa, and M.~d. G.~V. Nunes,
  ``A complex network approach to text summarization,'' \emph{Information
  Sciences}, vol. 179, no.~5, pp. 584--599, 2009.

\bibitem{sariyuce2018local}
A.~E. Sariy{\"u}ce, C.~Seshadhri, and A.~Pinar, ``Local algorithms for
  hierarchical dense subgraph discovery,'' \emph{Proc. of the VLDB Endowment},
  vol.~12, no.~1, pp. 43--56, 2018.

\bibitem{esfahani2019efficient}
F.~Esfahani, V.~Srinivasan, A.~Thomo, and K.~Wu, ``Efficient computation of
  probabilistic core decomposition at web-scale.'' in \emph{Proceedings of the
  22nd International Conference on Extending Database Technology (EDBT)}, 2019,
  pp. 325--336.

\bibitem{huang2016truss}
X.~Huang, W.~Lu, and L.~V. Lakshmanan, ``Truss decomposition of probabilistic
  graphs: Semantics and algorithms,'' in \emph{Proceedings of the 2016
  International Conference on Management of Data}.\hskip 1em plus 0.5em minus
  0.4em\relax ACM, 2016, pp. 77--90.

\bibitem{khaouid2015k}
W.~Khaouid, M.~Barsky, V.~Srinivasan, and A.~Thomo, ``K-core decomposition of
  large networks on a single pc,'' \emph{Proceedings of the VLDB Endowment},
  vol.~9, no.~1, pp. 13--23, 2015.

\bibitem{peng2018efficient}
Y.~Peng, Y.~Zhang, W.~Zhang, X.~Lin, and L.~Qin, ``Efficient probabilistic
  k-core computation on uncertain graphs,'' in \emph{Proceedings of the IEEE
  34th International Conference on Data Engineering (ICDE)}.\hskip 1em plus
  0.5em minus 0.4em\relax IEEE, 2018, pp. 1192--1203.

\bibitem{wang2012truss}
J.~Wang and J.~Cheng, ``Truss decomposition in massive networks,''
  \emph{Proceedings of the VLDB Endowment}, vol.~5, no.~9, 2012.

\bibitem{sariyuce2015finding}
A.~E. Sariy{\"u}ce, C.~Seshadhri, A.~Pinar, and U.~V. Catalyurek, ``Finding the
  hierarchy of dense subgraphs using nucleus decompositions,'' in
  \emph{Proceedings of the 24th International Conference on World Wide Web},
  2015, pp. 927--937.

\bibitem{sariyuce2017nucleus}
A.~E. Sariy{\"u}ce, C.~Seshadhri, A.~Pinar, and {\"U}.~V. {\c{C}}ataly{\"u}rek,
  ``Nucleus decompositions for identifying hierarchy of dense subgraphs,''
  \emph{ACM Transactions on the Web (TWEB)}, vol.~11, no.~3, pp. 1--27, 2017.

\bibitem{saxena2018social}
R.~Saxena, S.~Kaur, and V.~Bhatnagar, ``Social centrality using network
  hierarchy and community structure,'' \emph{Data Mining and Knowledge
  Discovery}, vol.~32, no.~5, pp. 1421--1443, 2018.

\bibitem{zhao2019effective}
Y.~Zhao, X.~Dong, and Y.~Yin, ``Effective and efficient dense subgraph query in
  large-scale social internet of things,'' \emph{IEEE Transactions on
  Industrial Informatics}, vol.~16, no.~4, pp. 2726--2736, 2019.

\bibitem{zhang2017hidden}
S.~Zhang, D.~Zhou, M.~Y. Yildirim, S.~Alcorn, J.~He, H.~Davulcu, and H.~Tong,
  ``Hidden: hierarchical dense subgraph detection with application to financial
  fraud detection,'' in \emph{Proceedings of the 2017 SIAM International
  Conference on Data Mining}.\hskip 1em plus 0.5em minus 0.4em\relax SIAM,
  2017, pp. 570--578.

\bibitem{wu2020extracting}
Q.~Wu, X.~Huang, A.~Culbreth, J.~Waltz, L.~E. Hong, and S.~Chen, ``Extracting
  brain disease-related connectome subgraphs by adaptive dense subgraph
  discovery,'' \emph{bioRxiv}, 2020.

\bibitem{ma2017detection}
X.~Ma, G.~Zhou, J.~Shang, J.~Wang, J.~Peng, and J.~Han, ``Detection of
  complexes in biological networks through diversified dense subgraph mining,''
  \emph{Journal of Computational Biology}, vol.~24, no.~9, pp. 923--941, 2017.

\bibitem{esfahani2020nucleus}
F.~Esfahani, V.~Srinivasan, A.~Thomo, and K.~Wu, ``Nucleus decomposition in
  probabilistic graphs: Hardness and algorithms,'' \emph{arXiv preprint
  arXiv:2006.01958}, 2020.

\bibitem{batagelj2011fast}
V.~Batagelj and M.~Zaver{\v{s}}nik, ``Fast algorithms for determining
  (generalized) core groups in social networks,'' \emph{Advances in Data
  Analysis and Classification}, vol.~5, no.~2, pp. 129--145, 2011.

\bibitem{le1960approximation}
L.~Le~Cam, ``An approximation theorem for the poisson binomial distribution.''
  \emph{Pacific Journal of Mathematics}, vol.~10, no.~4, pp. 1181--1197, 1960.

\bibitem{Lyapunov-Nouvelle}
A.~Lyapunov, ``Nouvelle forme de la th\'{e}oreme dur la limite de
  probabilit\'{e},'' \emph{M\'{e}moires de l’Academie Imp\'{e}riale des Sci.
  de St. Petersbourg}, vol.~12, pp. 1\--24, 1901.

\bibitem{haight1967handbook}
F.~A. Haight, ``Handbook of the poisson distribution,'' 1967.

\bibitem{rollin2007translated}
A.~R{\"o}llin, ``Translated poisson approximation using exchangeable pair
  couplings,'' \emph{The Annals of Applied Probability}, vol.~17, no. 5/6, pp.
  1596--1614, 2007.

\bibitem{papoulis2002probability}
A.~Papoulis and S.~U. Pillai, \emph{Probability, random variables, and
  stochastic processes}.\hskip 1em plus 0.5em minus 0.4em\relax Tata
  McGraw-Hill Education, 2002.

\bibitem{ehm1991binomial}
W.~Ehm, ``Binomial approximation to the poisson binomial distribution,''
  \emph{Statistics \& Probability Letters}, vol.~11, no.~1, pp. 7--16, 1991.

\bibitem{mukhopadhyay2000probability}
N.~Mukhopadhyay, \emph{Probability and statistical inference}.\hskip 1em plus
  0.5em minus 0.4em\relax CRC Press, 2000.

\bibitem{hoeffding1994probability}
W.~Hoeffding, ``Probability inequalities for sums of bounded random
  variables,'' in \emph{The Collected Works of Wassily Hoeffding}.\hskip 1em
  plus 0.5em minus 0.4em\relax Springer, 1994, pp. 409--426.

\bibitem{potamias2010k}
M.~Potamias, F.~Bonchi, A.~Gionis, and G.~Kollios, ``K-nearest neighbors in
  uncertain graphs,'' \emph{PVLDB}, vol.~3, no. 1-2, pp. 997--1008, 2010.

\bibitem{krogan2006global}
N.~J. Krogan, G.~Cagney, H.~Yu, G.~Zhong, X.~Guo, A.~Ignatchenko, J.~Li, S.~Pu,
  N.~Datta, A.~P. Tikuisis \emph{et~al.}, ``Global landscape of protein
  complexes in the yeast saccharomyces cerevisiae,'' \emph{Nature}, vol. 440,
  no. 7084, p. 637, 2006.

\bibitem{pfeiffer2011methods}
J.~J. Pfeiffer and J.~Neville, ``Methods to determine node centrality and
  clustering in graphs with uncertain structure,'' in \emph{Fifth International
  AAAI Conference on Weblogs and Social Media}, 2011.

\bibitem{blei2003latent}
D.~M. Blei, A.~Y. Ng, and M.~I. Jordan, ``Latent dirichlet allocation,''
  \emph{Journal of machine Learning research}, vol.~3, no. Jan, pp. 993--1022,
  2003.

\bibitem{biomine2019}
\BIBentryALTinterwordspacing
V.~Podpe$\check{c}$an, c.~Ram$\check{s}$ak, K.~Gruden, H.~Toivonen, and
  N.~Lavra$\check{c}$, ``Interactive exploration of heterogeneous biological
  networks with biomine explorer,'' \emph{Bioinformatics}, 06 2019. [Online].
  Available: \url{https://doi.org/10.1093/bioinformatics/btz509}
\BIBentrySTDinterwordspacing

\bibitem{bouhaddou2020global}
M.~Bouhaddou, D.~Memon, B.~Meyer, K.~M. White, V.~V. Rezelj, M.~C. Marrero,
  B.~J. Polacco, J.~E. Melnyk, S.~Ulferts, R.~M. Kaake \emph{et~al.}, ``The
  global phosphorylation landscape of sars-cov-2 infection,'' \emph{Cell}, vol.
  182, no.~3, pp. 685--712, 2020.

\bibitem{marchetti2020covid}
M.~Marchetti, ``Covid-19-driven endothelial damage: complement, hif-1, and abl2
  are potential pathways of damage and targets for cure,'' \emph{Annals of
  hematology}, pp. 1--7, 2020.

\bibitem{zheng2020examining}
W.-j. Zheng, Q.~Yan, Y.-s. Ni, S.-f. Zhan, L.-l. Yang, H.-f. Zhuang, X.-h. Liu,
  and Y.~Jiang, ``{Examining the effector mechanisms of Xuebijing injection on
  COVID-19 based on network pharmacology},'' \emph{BioData mining}, vol.~13,
  p.~17, 2020.

\bibitem{taniguchi2021increased}
K.~Taniguchi-Ponciano, E.~Vadillo, H.~Mayani, C.~R. Gonzalez-Bonilla,
  J.~Torres, A.~Majluf, G.~Flores-Padilla, N.~Wacher-Rodarte, J.~C. Galan,
  E.~Ferat-Osorio \emph{et~al.}, ``{Increased expression of hypoxia-induced
  factor 1$\alpha$ mRNA and its related genes in myeloid blood cells from
  critically ill COVID-19 patients},'' \emph{Annals of Medicine}, vol.~53,
  no.~1, pp. 197--207, 2021.

\bibitem{guo2021integrative}
Y.~Guo, F.~Esfahani, X.~Shao, V.~Srinivasan, A.~Thomo, L.~Xing, and X.~Zhang,
  ``{Integrative COVID-19 Biological Network Inference with Probabilistic Core
  Decomposition},'' \emph{Briefings in Bioinformatics}, 2021 (in press), doi:
  \href{http://dx.doi.org/10.1093/bib/bbab455}{10.1093/bib/bbab455},
  \url{https://www.biorxiv.org/content/10.1101/2021.06.23.449535v1.full.pdf}.

\bibitem{zhao2020imatinib}
H.~Zhao, M.~Mendenhall, and M.~W. Deininger, ``Imatinib is not a potent
  anti-sars-cov-2 drug,'' \emph{Leukemia}, vol.~34, no.~11, pp. 3085--3087,
  2020.

\bibitem{ebrahimi2021interferon}
K.~H. Ebrahimi, J.~Gilbert-Jaramillo, W.~S. James, and J.~S. McCullagh,
  ``Interferon-stimulated gene products as regulators of central carbon
  metabolism,'' \emph{The FEBS journal}, vol. 288, no.~12, p. 3715, 2021.

\bibitem{zhou2019metascape}
Y.~Zhou, B.~Zhou, L.~Pache, M.~Chang, A.~H. Khodabakhshi, O.~Tanaseichuk,
  C.~Benner, and S.~K. Chanda, ``Metascape provides a biologist-oriented
  resource for the analysis of systems-level datasets,'' \emph{Nature
  communications}, vol.~10, no.~1, pp. 1--10, 2019.

\bibitem{seidman1983network}
S.~B. Seidman, ``Network structure and minimum degree,'' \emph{Social
  networks}, vol.~5, no.~3, pp. 269--287, 1983.

\bibitem{huang2014querying}
X.~Huang, H.~Cheng, L.~Qin, W.~Tian, and J.~X. Yu, ``Querying k-truss community
  in large and dynamic graphs,'' in \emph{Proceedings of the 2014 ACM SIGMOD
  International Conference on Management of Data}, 2014, pp. 1311--1322.

\bibitem{cohen2008trusses}
J.~Cohen, ``Trusses: Cohesive subgraphs for social network analysis,''
  \emph{National security agency technical report}, vol.~16, pp. 3--1, 2008.

\bibitem{zhang2012extracting}
Y.~Zhang and S.~Parthasarathy, ``Extracting analyzing and visualizing triangle
  k-core motifs within networks,'' in \emph{Proceedings of IEEE 28th
  International Conference on Data Engineering}.\hskip 1em plus 0.5em minus
  0.4em\relax IEEE, 2012, pp. 1049--1060.

\bibitem{chen2014distributed}
P.~Chen, C.~K. Chou, and M.~Chen, ``Distributed algorithms for k-truss
  decomposition,'' in \emph{Proceedings of 2014 IEEE International Conference
  on Big Data (Big Data)}.\hskip 1em plus 0.5em minus 0.4em\relax IEEE, 2014,
  pp. 471--480.

\bibitem{chen2016efficient}
S.~Chen, R.~Wei, D.~Popova, and A.~Thomo, ``Efficient computation of importance
  based communities in web-scale networks using a single machine,'' in
  \emph{Proceedings of the 25th ACM International on Conference on Information
  and Knowledge Management}, 2016, pp. 1553--1562.

\bibitem{montresor2012distributed}
A.~Montresor, F.~De~Pellegrini, and D.~Miorandi, ``Distributed k-core
  decomposition,'' \emph{IEEE Transactions on parallel and distributed
  systems}, vol.~24, no.~2, pp. 288--300, 2012.

\bibitem{zhang2016engagement}
F.~Zhang, Y.~Zhang, L.~Qin, W.~Zhang, and X.~Lin, ``When engagement meets
  similarity: Efficient (k,r)-core computation on social networks,''
  \emph{Proc. VLDB Endow.}, vol.~10, no.~10, p. 998–1009, 2017.

\bibitem{cheng2011efficient}
J.~Cheng, Y.~Ke, S.~Chu, and M.~T. {\"O}zsu, ``Efficient core decomposition in
  massive networks,'' in \emph{Proceedings of the 2011 IEEE 27th International
  Conference on Data Engineering}.\hskip 1em plus 0.5em minus 0.4em\relax IEEE,
  2011, pp. 51--62.

\bibitem{sariyuce2013streaming}
A.~E. Sar{\'\i}y{\"u}ce, B.~Gedik, G.~Jacques-Silva, K.~L. Wu, and {\"U}.~V.
  {\c{C}}ataly{\"u}rek, ``Streaming algorithms for k-core decomposition,''
  \emph{Proceedings of the VLDB Endowment}, vol.~6, no.~6, pp. 433--444, 2013.

\bibitem{wang2020efficient}
K.~Wang, X.~Lin, L.~Qin, W.~Zhang, and Y.~Zhang, ``Efficient bitruss
  decomposition for large-scale bipartite graphs,'' in \emph{Proceedings of
  2020 IEEE 36th International Conference on Data Engineering (ICDE)}.\hskip
  1em plus 0.5em minus 0.4em\relax IEEE, 2020, pp. 661--672.

\bibitem{preti2021strud}
G.~Preti, G.~De~Francisci~Morales, and F.~Bonchi, ``Strud: Truss decomposition
  of simplicial complexes,'' in \emph{Proceedings of the Web Conference 2021},
  2021, pp. 3408--3418.

\bibitem{yang2019index}
B.~Yang, D.~Wen, L.~Qin, Y.~Zhang, L.~Chang, and R.~Li, ``Index-based optimal
  algorithm for computing k-cores in large uncertain graphs,'' in
  \emph{Proceedings of the IEEE 35th International Conference on Data
  Engineering (ICDE)}.\hskip 1em plus 0.5em minus 0.4em\relax IEEE, 2019, pp.
  64--75.

\bibitem{esfahani2019fast}
F.~Esfahani, J.~Wu, V.~Srinivasan, A.~Thomo, and K.~Wu, ``Fast truss
  decomposition in large-scale probabilistic graphs.'' in \emph{Proceedings of
  the 22nd International Conference on Extending Database Technology (EDBT)},
  2019, pp. 722--725.

\bibitem{sun2021efficient}
Z.~Sun, X.~Huang, J.~Xu, and F.~Bonchi, ``Efficient probabilistic truss
  indexing on uncertain graphs,'' in \emph{Proceedings of the Web Conference
  2021}, 2021, pp. 354--366.

\end{thebibliography}


\end{document}